\documentclass[12pt]{article}
\usepackage{graphicx,graphics,epsf,subfigure,epstopdf}
\usepackage{amsmath,amsthm,amsxtra,amsfonts,amscd,amssymb,bm}
\usepackage[round]{natbib}
\usepackage[bookmarks=true, colorlinks=true, linkcolor=blue, citecolor=blue, breaklinks=true]{hyperref}
\usepackage{times}
\usepackage{rotating}
\usepackage{algorithm,algorithmic}

\newcommand{\BaselG}{\mathrm{Basel}}
\newcommand{\Variance}{\mathrm{Variance}}
\newcommand{\VaR}{{\mathrm{VaR}_{\alpha}}}
\newcommand{\CVaR}{{\mathrm{CVaR}_{\alpha}}}
\newcommand{\BaselTwoDotFive}{\mathrm{Basel2.5}}

\newcommand{\BaselThree}{\mathrm{Basel3}}

\newcommand{\Lc}{\mathcal{L}}
\newcommand{\U}{\mathcal{U}}
\newcommand{\Ur}{\mathcal{U}_{r_0}}
\newcommand{\one}{\mathbf{1}}
\newcommand{\R}{\mathbb{R}}
\newcommand{\NU}{\mathcal{N}_{\Ur}}
\newcommand{\tR}{\tilde{R}}
\newcommand{\tY}{\tilde{Y}}

\newcommand{\be}{\begin{equation}}
\newcommand{\ee}{\end{equation}}
\newcommand{\bee}{\begin{equation*}}
\newcommand{\eee}{\end{equation*}}
\newcommand{\bea}{\begin{eqnarray}}
\newcommand{\eea}{\end{eqnarray}}
\newcommand{\beaa}{\begin{eqnarray*}}
\newcommand{\eeaa}{\end{eqnarray*}}

\usepackage[usenames]{color}
\usepackage[normalem]{ulem}
\definecolor{mygreen}{rgb}{0,.5,0}

\newtheorem{theorem}{Theorem}[section]
\newtheorem{proposition}{Proposition}[section]
\newtheorem{lemma}{Lemma}[section]

\newtheorem{definition}{Definition}[section]

\begin{document}

\title{ {Asset Allocation under the Basel Accord Risk Measures}\thanks{We are
grateful to Steven Kou for his insightful comments to the paper. Zaiwen Wen was
partially supported by the NSFC grant 11101274 and research fund (20110073120069) for the Doctoral
Program of Higher Education of China. Xianhua Peng was partially supported by a grant from School-Based-Initiatives of HKUST (Grant No.
SBI11SC03) and Hong Kong RGC Direct Allocation Grant (Project No. DAG12SC05-3). Xin Liu was partially supported by
the NSFC grants 10831006 and 11101409. Xiaodi Bai and Xiaoling Sun were supported
by the NSFC grant 10971034 and the Joint NSFC/RGC grant 71061160506.}}
\author{Zaiwen Wen\thanks{Department of Mathematics,  MOE-LSC and Institute of Natural Sciences, Shanghai
Jiaotong University, Shanghai 200240, China, zw2109@sjtu.edu.cn.} \and
Xianhua Peng\thanks{Department of Mathematics, Hong Kong University of Science and Technology, Hong Kong, maxhpeng@ust.hk.} \and
Xin Liu\thanks{Academy of Mathematics and Systems Science, Chinese Academy of Sciences,
Beijing 100080, China, liuxin@lsec.cc.ac.cn.} \and
Xiaodi Bai\thanks{Department of Management Science, School of Management, Fudan University,
Shanghai 200433, China, xdbai@fudan.edu.cn.} \and Xiaoling Sun\thanks{Department of Management Science, School of Management, Fudan University,
Shanghai 200433, China, xls@fudan.edu.cn.}}
\date{First version January 2013, this version August 2013}
%First version July 2006
\maketitle

\begin{abstract}
Financial institutions are currently required to meet more stringent capital requirements than they were before the recent financial crisis; in particular, the capital requirement for a large bank's trading book under the Basel
2.5 Accord more than doubles that under the Basel II Accord. The significant increase in capital requirements renders it necessary for banks to take into account the constraint of capital requirement when they make asset allocation decisions.
In this paper, we propose a new asset allocation model
that incorporates the regulatory capital requirements under both the Basel 2.5 Accord, which is currently in effect, and the Basel
III Accord, which was recently proposed and is currently under discussion.
We propose an unified algorithm based on the alternating
direction augmented Lagrangian method to solve the model; we also establish the first-order optimality of the limit points of the sequence generated by the algorithm under some mild conditions. The algorithm is simple and easy to implement; each step of the algorithm consists of solving convex quadratic programming or one-dimensional subproblems.
Numerical experiments on simulated and real market data show that the algorithm compares favorably with other existing methods, especially in cases in which the model is non-convex.\\

%and is able to find suboptimal solutions of good quality efficiently.

% In this paper, we propose asset allocation models
% in which the capital requirement calculated by the Basel 2.5 or Basel III
%risk measure, and  the risk level of the investment portfolio is measured by an
%internal risk measure which can be freely chosen by the portfolio manager.
%in the case of finite discrete distributions.
%Although our targeted model is essentially equivalent to mixed-integer programming problems,
%In the case that the model is non-convex,
% and essentially equivalent to mixed-integer programming
%problems

\end{abstract}

\section{Introduction}

One of the major consequences of the financial crisis that began in 2007 is that financial institutions are now required to meet more stringent capital requirements than they were before the crisis. The considerable increase in capital requirements has been imposed through the Basel Accords, which have  undergone substantial revision since the inception of the financial crisis. The framework of the latest version of the Basel Accord, the Basel III Accord \citep*{Basel-III}, was announced in December 2010 and is soon to be
implemented in many leading nations, including the United States \citep*{Fed-2012}.
%For example, the Board of Governors of the Federal Reserve Systems of the United States announced in December 2011 that it will implement the Basel III Accord \citep*{Fed-2012}.

In particular, the capital requirements for banks' trading books, which are calculated by the Basel Accord risk measure for the trading book, have been increased substantially.
%the Basel Accord risk measure for the trading book, which specifies how the capital requirements for banks' trading books should be calculated, has gone through major changes.
Before the 2007 financial crisis, the Basel II risk measure \citep*{Basel06} was
used in the calculation. During the crisis, it was found
that the Basel II risk measure had serious drawbacks, such as being procyclical
and not being conservative enough. In response to the financial crisis, the Basel committee revised the Basel II market
risk framework and imposed the ``Basel 2.5" risk measure \citep*{Basel09} in July 2009. It has been estimated that the capital requirement for a large bank's trading book under the Basel 2.5 risk measure on average \emph{more than doubles} that under the Basel II risk measure \citep[][p. 11]{BaselRev}.

The substantial increase in the capital requirements for the trading book makes it more important for banks to take into account the constraint of capital requirements when they construct investment portfolios. In this paper, we address this issue by proposing a new asset allocation model
that incorporates the capital requirement imposed by the Basal Accords. More precisely, we propose the \emph{``mean-$\rho$-Basel"} asset allocation model,
%The model incorporates both
%the risk measure $\rho$, which can be freely
%chosen by the portfolio manager for measuring the risk
%of the investment portfolio, such as  and the external risk measures, which are the Basel Accord risk measures used for setting capital constraints and are imposed by regulators.
in which $\rho$ denotes the risk measure used for measuring the risk of the investment portfolio, such as variance, value-at-risk (VaR), or conditional value-at-risk (CVaR); $\rho$ can be freely chosen by the portfolio manager; and ``Basel" denotes the constraint that the regulatory capital of the portfolio calculated by the Basel Accord risk measure should not exceed a certain upper limit.

The complexity of the Basel Accord risk measures for calculating the capital requirements
%2.5 and the recently proposed Basel Accord risk measure,
%which involve the calculation of VaR or CVaR under multiple scenarios including stressed scenarios,
poses a challenge to solving the proposed ``mean-$\rho$-Basel" asset allocation model.
%the problem of asset allocation with the Basel Accord capital requirements constraint.
%Hence, how to efficiently
%construct investment portfolios under the new constraints becomes a crucial
%issue for a bank.
%Basel III capital rule will have substantial impact on the way banks construct their portfolios to maximize their profit
%and to meet the new capital requirements.
%On the other hand, the Basel
%III risk measure is much more complicated than the previous ones.
%one of the most important banking regulation accords. The Basel Accords
%It represents a fundamental strengthening and a radical overhaul of global capital requirement standards, which is one of the critical parts of the global financial reform and reflects the dramatic impact of the recent financial crisis on the regulatory framework of the finance industry.
%Since 1996 until has been revised in 2009 and and will most likely be revised again due to a new proposal by the Basel committee.
%We also study a more general problem, i.e., the portfolio optimization problem under more general risk measures called natural risk statistics, which are a new class of risk measures recently proposed by \citet*{Kou-Peng-Heyde-09}. Natural risk statistics include the Basel III market risk measure for calculating the new capital requirements as special cases.
The Basel Accords use VaR or CVaR with \emph{scenario analysis} as the risk measure to calculate the capital requirements for a bank's trading book.
Scenario analysis is used to analyze the behavior of random losses under different scenarios; a scenario refers to a specific economic regime such as an economic boom and a financial crisis.
The Basel II risk measure involves the calculation of VaR under 60 different scenarios. The Basel 2.5 risk measure involves the calculation of VaR under 120 scenarios, including 60 stressed scenarios.
%In 2009, the Basel II risk measure was
%revised and changed to the Basel 2.5 risk measure \citep*{Basel09}, which
%incorporated the \emph{stressed} VaR, i.e., VaR calculated under stressed
%scenarios.
Most recently, in May 2012, the Basel Committee released a consultative document \citep*{BaselRev} that presents the initial policy proposal of a new risk measure to replace the Basel 2.5 risk measure for the trading book; the new risk measure involves the calculation of CVaR under stressed scenarios. Currently, this new proposal is under discussion and has not been finalized. It is beyond the scope of this paper to discuss whether the newly proposed risk measure is superior to the Basel 2.5 risk measure; hence, we will consider both the Basel 2.5 and the newly proposed Basel risk measure in the mean-$\rho$-Basel model.
% regarding the Basel committee's fundamental review of the trading book capital requirements. In the document, the committee proposes to revise the Basel 2.5 trading book risk measure by replacing VaR by ES (or, equivalently, CVaR).
See Section \ref{subsec:basel_rm} for details regarding the Basel Accord risk measures.

Numerous studies have examined the single-period
asset allocation model of \emph{``mean-$\rho$"}, in which $\rho$ is a measure of portfolio
risk such as variance, VaR, or CVaR. On recent developments in the mean-variance asset
allocation models and associated algorithms, see e.g., \citet{CKTVanRoy-2012}. \cite{Iyengar-Ma-2013}
propose a fast iterative gradient descent algorithm capable of handling large-scale problems for the mean-CVaR problem. \cite*{Lim2011163} evaluate CVaR as the risk measure in data-driven portfolio
optimization and show that portfolios obtained by solving mean-CVaR problems are unreliable
due to estimation errors of CVaR and/or the mean asset returns. To address the issue of estimation
risk, \citet*{KLV-2011} introduce a new approach, called performance-based regularization,
to the data-driven mean-CVaR portfolio optimization problem.
%In a recent interesting paper, \citet*{CKTVanRoy-2012} formulate a mean-variance portfolio selection problem that generalizes the Black-Litterman model and accommodates qualitative input about expected returns; see the reference therein for an overview of the mean-variance problem.
\citet*{Rock02} develop a method to reduce the data-driven mean-CVaR asset allocation problem
to a linear programming (LP) problem. The mean-VaR problem is more difficult than the mean-CVaR due to the non-convexity of VaR.
%it is NP-hard and is equivalent to a mixed integer LP problem under finite discrete distribution \citep*{Benati-Rizzi-2005}.
Software packages
such as CPLEX can be used to solve small-to-medium sized problems of this type.
%\citet*{Larsen-Mausser-Uryasev-2002} proposed two heuristic algorithms for the mean-VaR problem based on the LP algorithm for solving the mean-CVaR problem.
%\citet*{Benati-Rizzi-2005} showed that the problem is NP-hard %and formulated it as a mixed integer linear programming %problem, which allows them to solve medium size problem using CPLEX.
Recently, \cite{cui2013nonlinear} propose a second-order cone programming method
to solve a mean-VaR model when VaR is estimated by its first-order or  second-order
approximations.
 \citet{BaiZhengSunSun2012} propose a penalty decomposition method for
probabilistically constrained programs including the mean-VaR
problem.
%\citet*{HongYangZhang2011} developed a novel sequential convex approximation for general chance-constrained convex
%programs by representing the non-convex probability function as the difference of two convex functions. \citet*{Natarajan-Pachamanova-Sim-2008}
%proposed a computationally tractable approximation method for minimizing
%VaR based on robust optimization techniques. The method leads to the optimization of a modified VaR measure called Asymmetry-Robust VaR. \citet*{Gaivoronski-Pflug-2005} described a method based on approximating
%VaR by a smoothed VaR function which filters out local irregularities.

%\citet*{Wozabal-Hochreiter-Pflug-2010} formulated the VaR constrained problem as
%a D.C. (difference of
%two convex functions) optimization problem and
%used a conical Branch-and-Bound algorithm for the resulting D.C. problem. % to find the global optimal solution. %\citet*{Wozabal-2010} applied the difference of convex algorithm (DCA), which is an approximate algorithm but is more computationally tractable, to solve the problem.

%However, the method is not tractable for large scale problems.
%\citet*{Chiarawongse2012,LiZhouLim2002}

It appears to be more challenging to solve the mean-$\rho$-Basel model than the mean-$\rho$ model due to the complexity of the Basel Accord risk measures that involve multiple VaRs or CVaRs under various scenarios and the non-convexity of VaR. In this paper, we develop an unified and computationally efficient method to solve the mean-$\rho$-Basel problem.
%We first decompose the risk measures into simpler structures by
%introducing certain bridging variables, and then  apply the alternating direction method (ADM or block coordinate descent scheme) to minimize
%the augmented Lagrangian functions for the transformed problem.
This method is based on the \emph{alternating direction augmented Lagrangian method} (ADM)  (see, e.g., \citealt*{WGY-2010}; \citealt{HeYuan2012}; \citealt*{HongLuo2012} and the references therein). %, or, equivalently, block coordinate descent scheme.
The method is very simple and easy to implement; it
reduces the
original problem to one-dimensional optimization or convex quadratic programming subproblems
that may even have closed-form solutions; hence, the method is capable of solving large scale problems. When the mean-$\rho$-Basel problem is convex for some specific $\rho$ and Basel constraint (e.g., $\rho$ is variance and the Basel constraint is specified by the newly proposed Basel Accord risk measure), the method is guaranteed to converge to the globally optimal solution; when the problem is non-convex,
we show that the limit points of the sequence generated by the method satisfy the first-order optimality condition. See Section \ref{sec:adm_algo}.
%The algorithmic difference for mean-VaR, mean-CVaR and mean-Basel
% under the ADM framework lies in the  subproblem corresponding to specific risk
% measures,
% respectively.
%The convergence properties of the
%method are established under mild conditions.
%To evaluate the
%performance of our method, we conduct computational , in particular, those from the financial crisis
%in 2008.

%Comparisons with the penalty decomposition method  in \citet*{BaiZhengSunSun2012}  indicate that our method is better in terms of the quality of the suboptimal solutions and the computational time.

The proposed method also applies to mean-$\rho$ problems such as the mean-VaR, mean-CVaR, and \emph{``mean-Basel"} problem, in which the Basel Accord risk measures are used to quantify the risk of the portfolio. The Basel Accord risk measures involve multiple VaR or CVaR under different scenarios, which essentially correspond to different models or distributions of  asset returns. Hence, using the Basel Accord risk measures, or, more generally, VaR or CVaR with scenario analysis, as the portfolio risk measure provides a way to address the problem of model uncertainty.

In summary, the main contribution of the paper is two-fold.  (i)   We formulate a new asset
allocation model, the mean-$\rho$-Basel model, which takes into account the regulatory capital constraint specified by the Basel Accord risk measure for trading books. We also formulate and study the related mean-Basel model. To the best of our
knowledge, there has been no literature on asset allocation involving the Basel Accord risk measures. (ii)
We propose an efficient alternating direction augmented Lagrangian method for solving
the mean-$\rho$-Basel and mean-$\rho$ models.
%Although the ADM approach has been widely
%used in convex optimization \citep{HeLiaoHanYang2002, WGY-2010, YangZhang2009},
%it appears that its use for
%solving optimization problems involving VaR or Basel Accords risk measures is new.
%Our method is different from the penalty
%decomposition methods in \citet{BaiZhengSunSun2012} as their division of blocks of variables leads to subproblems that
%are more expensive to solve.
For non-convex cases of these models, we establish the first-order optimality of the limit points of iterative sequence generated by the method under mild conditions. Although
there is no theoretical guarantee that the method will converge to the global
solution in non-convex cases of these models, numerical experiments on simulated and real market data
show that
the method can identify suboptimal
solutions that can often be superior to the approximate solutions of the mixed-integer programming formulation computed by CPLEX within one hour.

%Although the penalty
%decomposition methods in \citet{BaiZhengSunSun2012} also use the ADM framework, our proposed method is different from theirs and is more efficient because
% their division of blocks of variables leads to subproblems that
%are more expensive to solve, and numerical experiments showed that the number of iteration required in our method
%%the augmented Lagrangian framework
%may be smaller than that of theirs.
%% quadratic penalty case.

The remainder of the paper is organized as follows. In Section
\ref{sec:risk_measures}, we review the definition and properties of the Basel Accord
risk measures for trading books as well as some other relevant risk measures. In Section
\ref{sec:opt_models}, we
%discuss the data-based Basel Accord risk measures and
formulate the mean-$\rho$-Basel asset allocation model in which the Basel Accord risk measures are used for setting a regulatory capital constraint. In Section \ref{sec:adm_algo},
we propose the alternating direction augmented Lagrangian method for solving
the mean-$\rho$-Basel and the mean-$\rho$ problems; we also
provide convergence analysis of the method. Section
\ref{sec:numerical} provides the numerical results,  which demonstrate the
accuracy and efficiency of the proposed method. %Section \ref{sec:conclusion} concludes.

%\section{Asset Allocation using the Mean-VaR Formulation}
\section{Review of Relevant Risk Measures}\label{sec:risk_measures}
Variance is probably the best-known risk measure; in addition
to variance, there is a vast literature on theoretical frameworks
and concrete examples of risk measures. As it is beyond the scope of
this paper to discuss and compare different risk measures, we review
only the risk measures that are used in the asset allocation
problems considered in this paper.

% are most relevant such as Value-at-Risk, conditional Value-at-Risk, and the Basel Accord risk measures.

\subsection{Value-at-Risk and Conditional Value-at-Risk (Expected Shortfall)}

\emph{Value-at-Risk} (VaR) is one of the most widely used risk measures in risk management. VaR is a quantile of the loss distribution at
some pre-defined probability level. More precisely, let $F_X(\cdot)$ be the distribution
function of the random loss $X$, then, for a given $\alpha
\in (0,1)$, $\text{VaR}$ of $X$ at level $\alpha$ is defined as
\be \label{equ:var}
\text{VaR}_{\alpha}(X):=\inf \{x\mid F_X(x)\geq \alpha \}=F_X^{-1}(\alpha).
\ee
%\footnote{VaR has two well-known properties: ordinal
%covariance and monotonicity with respect to first order stochastic
%dominance \citep[see, e.g.,][]{Denneberg94}. \citet{Chambers2009}
%shows that the two properties essentially characterize VaR.}
%Given a data set $x=(x_1,\ldots,x_n)\in \mathbb{R}^n$ denoting the observations of $X$, $\text{VaR}_{\alpha}(X)$
%can be estimated by the sample quantile $x_{(\lceil \alpha n\rceil)}$,
%which is a
%strongly consistent and asymptotically normal estimator for $\text{VaR}_{\alpha}(X)$ \citep{Shao03}.
%, and is a special case of L-statistics, i.e., linear functions of order statistics.
%Basel II \citep{Basel06} allows institutions to use an internal ratings based approach to calculate capital charges
%for credit risk. The internal ratings based capital charge on a given instrument depends only on its own
%characteristics, and not those of the portfolio in which it is held. \citet{Gordy-2003} demonstrates that
%under certain conditions the rating based capital charge is asymptotically equal to the 99.9\% VaR of the portfolio loss.
%\citet{Duffie-Pan-97}, \citet{Gordy-2003}, %, and %\citet{Hull09}
\citet{Jorion07} provides a comprehensive discussion of
VaR and risk management.

\emph{Conditional Value-at-Risk} (CVaR), proposed by \citet{Rock02}, is another prominent and widely used risk measure. For the random loss $X$, the $\alpha$-tail distribution function of $X$ is defined as
\begin{equation}\label{equ:alpha_tail_dis}
  F_{\alpha,X}(x):=\begin{cases}
    0, & \text{for}\ x < \text{VaR}_{\alpha}(X),\\
    \frac{F_X(x)-\alpha}{1-\alpha}, & \text{for}\ x \geq \text{VaR}_{\alpha}(X).
  \end{cases}
\end{equation}
Then, the CVaR at level $\alpha$ of $X$ is defined as
\begin{align}\label{equ:cvar_def_ger_dis}
  \text{CVaR}_{\alpha}(X):=\text{mean of the}\ \alpha\text{-tail distribution of}\ X
 = \int_{-\infty}^{\infty}xdF_{\alpha,X}(x).
\end{align}
\emph{Expected shortfall (ES)} is a risk measure that is equivalent to CVaR and that was introduced independently in \citet*{AcerbiTachi}. CVaR and ES have the subadditivity property and belong to the class of coherent risk measures \citep{Artz99}; VaR
may not satisfy subadditivity and belongs to another class of risk measures called insurance risk measures \citep*{Wang97}.

\subsection{Basel Accord Risk Measures for Trading Books}\label{subsec:basel_rm}

%Basel Accords use VaR or CVaR with scenario analysis as the risk measure for calculating the capital requirements for a bank's trading book. A scenario refers to a specific economic regime such as an economic boom and a financial crisis; scenario analysis is the approach to analyze the behavior of the random loss under different scenarios. Scenario analysis is necessary because studies have shown that behavior of economic variables is substantially different
%at different regimes of economy \citep[see, e.g.,][]{Hamilton-1989}. In particular, many economic variables exhibit
%dramatic changes in their behavior during financial crisis \citep{Hamilton-2005}
%or when government monetary or fiscal policy undergo sudden changes \citep{Hamilton-1988, Sims-Zha-2004}.
%The abrupt changes in economic variables also show up in asset prices, which is manifested by the increase
%in volatility and correlation among asset returns in economic downturn \citep[see, e.g.,][]{Ang-Bekaert-2002,
%Dai-Singleton-Yang-2007}.

The Basel Accords use VaR or CVaR with scenario analysis as the risk measure for calculating capital requirements for banks' trading books. A scenario refers to a specific economic regime, such as an economic boom or a financial crisis. Scenario analysis is necessary because studies have shown that the behavior of economic variables is substantially different under different economic regimes \citep[see, e.g.,][]{Hamilton-1989}. In particular, many economic variables exhibit
dramatic changes in their behavior during financial crises \citep{Hamilton-2005}
or when government monetary or fiscal policies undergo sudden changes \citep{Sims-Zha-2006}.
There is also evidence that
% abrupt changes in economic variables also show up in asset prices, which is manifested by the increase
 the volatility and correlation among asset returns increase in economic downturns \citep*[see, e.g.,][]{Dai-Singleton-Yang-2007}.

The Basel II Accord \citep*{Basel06} specifies that the capital charge for the trading book on any
particular day $t$ for banks using the internal models approach should be calculated by the formula
\begin{equation}\label{equ:basel_2_mkt_risk}
c_t=\max\left\{\text{VaR}_{\alpha, t-1}(X),\frac{k}{60}\sum_{s=1}^{60} % was i=s before
\text{VaR}_{\alpha, t-s}(X)\right\},
\end{equation}
where $X$ is the loss of the bank's trading book; $k$ is a constant that is no less than 3; $\text{VaR}_{\alpha, t-s}(X)$
is the 10-day VaR of $X$ at $\alpha=99\%$ confidence level calculated on day $t-s$,
$s=1,\ldots,60$. $\text{VaR}_{\alpha, t-s}(X)$ is calculated under the scenario corresponding to information available on day $t-s$.
For example, $\text{VaR}_{\alpha, t-s}(X)$ of a portfolio of equity options is calculated conditional on the value of
the equity prices, equity volatilities, yield curves, etc., on day $t-s$. Therefore, the Basel II risk measure is a VaR with scenario analysis that involves 60 scenarios. %$\text{VaR}_{\alpha, t-s}(X)$ can be calculated based on a %model for the portfolio loss or be estimated from data sets %generated by historical simulation or Monte Carlo
%simulation \citep*{Jorion07, Hull09}.
%) $x^{[s]}=(x^{[s]}_1,x^{[s]}_2,\ldots,x^{[s]}_{n_{s}})\in\mathbb{R}^{n_{s}}$,
%where $n_s$ is the data size.

%The Basel III accord basically aims to improve the quality of banks' capital and significantly increase the level of banks' capital. To achieve this goal, the Basel III accord has significantly increased the capital requirements for market risk.

Since the 2007 financial crisis, the Basel II risk measure
\eqref{equ:basel_2_mkt_risk} has been criticized for two reasons: (i) This risk measure is based on contemporaneous observations and hence is  \emph{procyclical}, i.e.,
risk measurement obtained by it tend to be low in booms and high in crises, which is exactly opposite to the goal of
effective regulation \citep{Adrian2008}.
%\citet{BCGPS2009}
%provide in-depth discussion of the fundamental principles of financial
%regulation and emphasize the central importance of making capital requirements
%counter-cyclical.
(ii) This risk measure is not conservative enough. In fact, banks' actual losses during the financial crisis were significantly higher than the capital requirements calculated by the risk measure.

%Scenario analysis can help to reduce the procyclicality by using not only contemporaneous observations
%but also data under distressed scenarios that capture rare tail events which could cause severe losses. Indeed,
In response to the financial crisis, the Basel Committee revised the Basel II market
risk framework and replaced the Basel II risk measure with the ``Basel 2.5" risk measure in July 2009 \citep*{Basel09}. The Basel 2.5 risk measure for calculating capital requirements
for trading books is defined by
\begin{align}\label{equ:basel_2_rev}
c_t=&\max\left\{\text{VaR}_{\alpha, t-1}(X),\frac{k}{60}\sum_{s=1}^{60}
\text{VaR}_{\alpha, t-s}(X)\right\}\notag\\
&+\max\left\{\text{sVaR}_{\alpha, t-1}(X),\frac{\ell}{60}\sum_{s=1}^{60}
\text{sVaR}_{\alpha, t-s}(X)\right\},
\end{align}
where $\text{VaR}_{\alpha, t-s}(X)$ is the same as that in
\eqref{equ:basel_2_mkt_risk}; $k$ and $\ell$ are constants no less than 3; and $\text{sVaR}_{\alpha, t-s}(X)$ is
called the \emph{stressed} VaR of $X$ on day $t-s$ at confidence level $\alpha=99\%$, which is calculated
under a scenario in which the financial market is under significant stress, such
as the one that happened during the period from 2007 to 2008. The additional capital requirements based on stressed VaR help to reduce the procyclicality of the Basel II risk measure \eqref{equ:basel_2_mkt_risk} and significantly increase the capital requirements.

In May 2012, the Basel Committee released a consultative document \citep{BaselRev} that presents the initial policy proposal regarding the Basel Committee's fundamental review of the trading book capital requirements. In particular, the Committee proposed a new risk measure to replace the Basel 2.5 risk measure; the new risk measure uses CVaR (or, equivalently, ES) instead of VaR to calculate capital requirements. More precisely, under the new risk measure, the capital requirement for a group of trading desks that share similar major risk factors, such as equity, credit, interest rate, and currency, is defined as the CVaR of the loss that may be incurred by the group of trading desks; the CVaR should be calculated under stressed scenarios rather than under current market conditions. For example, an equity trading desk and an equity option trading desk would be grouped together for the purpose of calculating regulatory capital. This proposed risk measure is currently under discussion, and it is not yet clear whether it is going to be the final version of the Basel III risk measure. In addition, the proposal has not clearly stated if the capital charge for the $t$th day will depend solely on the stressed CVaR calculated on day $t-1$ or on the CVaR calculated on day $t-s$ for $s=1, 2, \ldots, 60$, as in Basel 2.5. To be more consistent with Basel 2.5, we consider the following
% Nonetheless, it is worthwhile studying the following
%the proposed new risk measure involves the calculation of CVaR under only one stressed scenario, it might be desirable to include the calculation of CVaR under more stressed scenarios, just as in Basel 2.5.
%Hence, with some abuse of the terminology and for the sake of notational convenience, we will call the following risk measure the
``Basel III" risk measure:
%More precisely, the Committee propose that the capital requirements for a
%proposed the Basel III risk measure for setting capital requirement for banks' trading book , which replaces VaR in the Basel 2.5 risk measure by ES (or equivalently CVaR).
%More precisely, the proposed Basel III risk measure is defined as
%One can obtain the following \emph{Basel-CVaR} risk measure by replacing the VaR :
%\begin{equation}\label{equ:basel_cvar}
%c_t=\max\left\{\text{CVaR}_{\alpha,t-1},\frac{k}{60}\sum_{s=1}^{60}
%\text{CVaR}_{\alpha,t-s}\right\}+\max\left\{\text{sCVaR}_{\alpha, t-1},\frac{\ell}{60}\sum_{s=1}^{60}
%\text{sCVaR}_{\alpha, t-s}\right\},
%\end{equation}
\begin{equation}\label{equ:basel_cvar}
c_t=\max\left\{\text{sCVaR}_{\alpha, t-1},\frac{\ell}{60}\sum_{s=1}^{60}
\text{sCVaR}_{\alpha, t-s}\right\},
\end{equation}
where
$\text{sCVaR}_{\alpha, t-s}$ is the \emph{stressed} CVaR at level $\alpha$ calculated on day $t-s$. The proposal suggests specifying $\alpha$ to be a level smaller than 99\% due to the difficulty of estimating CVaR at high confidence levels, but the exact value of $\alpha$ has not been determined. In the numerical examples of Section \ref{sec:numerical}, we choose $\alpha=98\%$.

\section{A New Asset Allocation Model Incorporating the Basel Accord Capital Constraint}\label{sec:opt_models}

Consider a portfolio composed of $d$ assets and let $u=(u_1, u_2, \ldots, u_d)^\top\in \mathbb{R}^{d}$ denote the portfolio weights of these assets, which are the percentage of initial wealth invested in the assets. Let $R=(R_1, R_2, \ldots,
R_d)^\top\in \mathbb{R}^{d}$ be the random vector of simple returns of these assets over a specified time horizon, e.g., one day. Then the simple return of the portfolio is $R^\top u$ and $-R^\top u$ is the loss of the portfolio (per \$1 of investment). Let $\mu\in \mathbb{R}^{d}$ be the (estimated) expected returns of the $d$ assets. Then $\mu^\top
u$ is the expected return of the portfolio.

The risk of the portfolio is measured by $\rho(-R^\top u)$, where $\rho$ is a properly chosen risk measure.
There are generally two approaches to the computation of $\rho(-R^\top u)$: (i) one
first assumes and estimates a (parametric) probability model for the joint distribution of $R$ and then
computes $\rho(-R^\top u)$; (ii) one estimates the risk $\rho(-R^\top u)$ directly from the historical observations of $R$ without assuming any hypothetical model for $R$.

As discussed in Section \ref{subsec:basel_rm}, the return vector $R$ is usually observed under different scenarios, such as economic booms and financial crises. Suppose there are $m$ scenarios. For each $s=1, \ldots,
m$, let $\tR^{[s]}\in \mathbb{R}^{n_s \times d}$ be the
collection of $n_s$ observations of $R$ under the $s$th scenario,
where each row of $\tR^{[s]}$ represents one observation of $R^\top$.
%Define the matrix
%%$\tR:=(\tR^{[1]}; \tR^{[2]}; \ldots; \tR^{[m]})\in \mathbb{R}^{n\times d}$.
%$\tR:=((\tR^{[1]})^\top,  (\tR^{[2]})^\top,  \ldots, (\tR^{[m]})^\top)^\top\in \mathbb{R}^{n\times d}$.
% %Then $- \tR^{[s]}u$ are the observations of the portfolio loss under the $s$th scenario, $s=1, 2, \ldots, m$. For a given $\alpha \in (0,1)$, let $p_s=\lceil\alpha n_s\rceil$, $s=1, \ldots,m$.
%Then the observations of the portfolio loss are given by
%\begin{equation}\label{equ:port_loss_obs}
%  %x(u):=-\tR u=(x^{[1]}(u); x^{[2]}(u); \ldots; x^{[m]}(u))\in \mathbb{R}^{n },
%  x(u):=-\tR u=(x^{[1]}(u)^\top, x^{[2]}(u)^\top, \ldots, x^{[m]}(u)^\top)^\top\in \mathbb{R}^{n },
%\end{equation}
Then, we define the matrix $\tR$ and the observations of portfolio loss $x(u)$ as
follows:
\begin{equation}\label{equ:port_loss_obs}
  \tR:=\begin{pmatrix}
    \tR^{[1]} \\ \tR^{[2]} \\ \vdots \\ \tR^{[m]}
  \end{pmatrix}
  \in \mathbb{R}^{n\times d},\ x(u):=-\tR u=\begin{pmatrix}
  -\tR^{[1]} u \\ -\tR^{[2]} u \\ \vdots \\ -\tR^{[m]} u
\end{pmatrix}=
  \begin{pmatrix}
  x^{[1]}(u) \\ x^{[2]}(u) \\ \vdots \\ x^{[m]}(u)
\end{pmatrix} \in \mathbb{R}^{n},\ n:=\sum_{s=1}^m n_s,
 \end{equation}
where
$  x^{[s]}(u):=-\tR^{[s]} u\in \mathbb{R}^{n_s}$
denotes the observations of portfolio loss under the $s$th scenario, $s=1,
2, \ldots, m$.

In this paper, we estimate $\rho(-R^\top u)$ directly from the return observations $\tR$, as this approach does not require a subjective model for $R$ and hence greatly reduces model misspecification error.

\subsection{Sample Versions of Measures of Portfolio Risk}\label{subsec:risk_stat_def}

In the following, we use $\lceil \cdot \rceil$ to denote the ceiling function. For $x=(x_1,
x_2, \ldots, x_n)^\top\in \R^n$,
let $(i_1, i_2, \ldots, i_n)$ be a permutation of $(1, 2, \ldots, n)$ such that $x_{i_1}\le x_{i_2} \le \cdots \le x_{i_n}$. Then, we define $x_{(j)}:= x_{i_j}$, $j=1, \ldots, n$; hence, $x_{(j)}$ denotes the $j$th smallest component of $x$.

Given the observation $\tR^{[s]}$, the empirical distribution function
of $(-R^{\top}u)$ under scenario $s$ is given by
\be\label{equ:emp_dis}
\hat F^{[s]}_{(-R^{\top}u)}(y):=\frac{1}{n_s}\sum_{i=1}^{n_s}1_{\{x^{[s]}(u)_i\leq y\}}.
\ee
Then, for each risk measure $\rho$ discussed in Section \ref{subsec:basel_rm}, $\rho(-R^{\top} u)$ can be estimated from the return observations $\tR$ by substituting $\hat F^{[s]}_{(-R^{\top}u)}(\cdot)$ for the distribution function of $(-R^{\top}u)$ under each scenario $s$. Thus, we obtain the following sample versions of risk measures.

\begin{description}
%\subsubsection{Mean-VaR}
\item[Variance:] Suppose there is one scenario, i.e., $m=1$. Then the sample variance of portfolio return is
\be \label{eq:rho-var} \rho_{\Variance}(x(u))=\frac{1}{n}x(u)^\top x(u)-\frac{1}{n^2}x(u)^{\top}\one\one^{\top}x(u),\ \text{where}\ \one:=(1, 1, \ldots, 1)^\top\in \mathbb{R}^{n}. \ee

\item[VaR:] Suppose that $m=1$. For a given $\alpha \in
(0,1)$, let $p=\lceil\alpha n\rceil$. %denotes the smallest integer that is greater than or equal to $x$.
Then the sample VaR at level $\alpha$ of the
portfolio is
\be \label{equ:VaR} \rho_{\VaR}(x(u))
:= x(u)_{(p)}=(-\tR u)_{(p)}. \ee

\item[CVaR:] Suppose that $m=1$. For a given $\alpha \in
(0,1)$, let $p=\lceil\alpha n\rceil$. Then the sample CVaR at level $\alpha$ of the
portfolio is
\be \label{equ:e_cvar}
\rho_{\CVaR}(x(u))
:= \frac{p-\alpha n}{(1-\alpha)n}x(u)_{(p)}+\frac{1}{(1-\alpha)n}\sum_{i=p+1}^n x(u)_{(i)}. \ee
By Theorem 10 in \cite{Rock02}, $\rho_{\CVaR}(x(u))$ can also be represented by
\be\label{equ:cvar_opt_rep}
\rho_{\CVaR}(x(u))=\min_{t\in\mathbb{R}} \;\; t+\frac{1}{(1-\alpha)n}\sum_{i=1}^{n}(x(u)_i-t)_+,\ \text{where}\ y_+:=\max(y, 0).
\ee

\item[Basel 2.5:] For a given $\alpha \in (0,1)$, let $p_s=\lceil\alpha n_s\rceil$, $s=1, \ldots,m$. Then $x^{[s]}(u)_{(p_s)}$ is the sample VaR at level $\alpha$ of the portfolio estimated from the data set $\tR^{[s]}$. Let $m_1=m_2=60$ and $m=120$. Suppose the first $m_1$ scenarios correspond to current market conditions and the last $m_2$ scenarios correspond to stressed scenarios. Then, the sample version of the Basel 2.5 Accord risk measure is given by
\begin{align}\label{equ:basel_3_sta}
   \rho_{\BaselTwoDotFive}(x(u)):= &\max\left\{ x^{[1]}(u)_{(p_1)},
  \frac{k}{m_1}\sum_{s=1}^{m_1}x^{[s]}(u)_{(p_s)}\right\}\notag\\
  & +\max\left\{x^{[m_1+1]}(u)_{(p_{m_1+1})},
  \frac{\ell}{m_2}\sum_{s=m_1+1}^{m}x^{[s]}(u)_{(p_s)}\right\}.
\end{align}
%Because the two constants $k$ and $\ell$ in the above formula are no less than 3, one would expect that $x^{[1]}(u)_{(p_1)}\leq
%  \frac{k}{m_1}\sum_{s=1}^{m_1}x^{[s]}(u)_{(p_s)}$ and $x^{[m_1+1]}(u)_{(p_{m_1+1})}\leq
%  \frac{\ell}{m_2}\sum_{s=m_1+1}^{m}x^{[s]}(u)_{(p_s)}$
%  in many situations. Hence, in many cases $\rho_{\BaselTwoDotFive}(x(u))$ would be equal to the following ``simplified Basel 2.5" risk measure, which is defined as
%\bea\label{equ:basel_3_sta_ave}
%\rho_{\BaselTwoDotFiveS}(x(u)) &:=& \frac{k}{m_1}\sum_{s=1}^{m_1}x^{[s]}(u)_{(p_s)}+ \frac{\ell}{m_2}\sum_{s=m_1+1}^{m}x^{[s]}(u)_{(p_s)}. \nonumber
%\eea

\item[Basel III:] Let $\alpha$ and $p_s$ be defined previously. Then
\be \label{equ:e_s_cvar}
\rho_{\CVaR}(x^{[s]}(u))
:= \frac{p_s-\alpha n_s}{(1-\alpha)n_s}x^{[s]}(u)_{(p_s)}+\frac{1}{(1-\alpha)n_s}\sum_{i=p_s+1}^{n_s}x^{[s]}(u)_{(i)} \ee
is the sample CVaR at level $\alpha$ of the portfolio estimated from the data set $\tR^{[s]}$. Suppose the first $m_1=60$ scenarios correspond to current market conditions and the last $m_2=60$ scenarios correspond to stressed scenarios. Then the sample version of the Basel-III risk measure is
%\bea\label{equ:basel_3_cvar}
%      \rho_{\BaselThree}(x(u))
%   &:= &\max\left\{ \rho_{\CVaR}(x^{[1]}(u)),
%  \frac{k}{m_1}\sum_{s=1}^{m_1}\rho_{\CVaR}(x^{[s]}(u))\right\}\\
%  &&+\max\left\{\rho_{\CVaR}(x^{[m_1+1]}(u)),
%  \frac{\ell}{m_2}\sum_{s=m_1+1}^{m}\rho_{\CVaR}(x^{[s]}(u))\right\}, \nonumber
%\eea
%where $m_1=m_2=60$ and $m=120$. Similarly, the simplified Basel III risk
%measure can be defined:
%\bea\label{equ:basel_3_cvar-a}
%     \rho_{\BaselThreeS}(x(u)) :=
%     \frac{k}{m_1}\sum_{s=1}^{m_1}\rho_{\CVaR}(x^{[s]}(u)) +
%     \frac{\ell}{m_2}\sum_{s=m_1+1}^{m}\rho_{\CVaR}(x^{[s]}(u)).
%\eea
\bea\label{equ:basel_3_cvar}
      \rho_{\BaselThree}(x(u)):=\max\left\{\rho_{\CVaR}(x^{[m_1+1]}(u)),
  \frac{\ell}{m_2}\sum_{s=m_1+1}^{m}\rho_{\CVaR}(x^{[s]}(u))\right\}.
\eea
%Similarly, the simplified Basel III risk
%measure can be defined:
%%\bea\label{equ:basel_3_cvar-a}
%%     \rho_{\BaselThreeS}(x(u)) :=
%%     \frac{k}{m_1}\sum_{s=1}^{m_1}\rho_{\CVaR}(x^{[s]}(u)) +
%%     \frac{\ell}{m_2}\sum_{s=m_1+1}^{m}\rho_{\CVaR}(x^{[s]}(u)).
%%\eea
%\bea\label{equ:basel_3_cvar-a}
%     \rho_{\BaselThreeS}(x(u)) :=
%     \frac{\ell}{m_2}\sum_{s=m_1+1}^{m}\rho_{\CVaR}(x^{[s]}(u)).
%\eea

%\item[Basel-TTCE:] For any given $0 < \alpha < \beta < 1$, let $p_s=\lceil\alpha n_s\rceil$ and $q_s=\lceil\beta n_s\rceil$, $s=1, \ldots,m$. Then
%\begin{equation}\label{equ:trunc_tce_est}
%  \rho_{\TTCE}(x^{[s]}(u)):=\frac{1}{q_s - p_s}\sum_{i=p_s+1}^{q_s} x^{[s]}(u)_{(i)}
%\end{equation}
%is an estimate of the $\text{TTCE}_{\alpha,\beta}$ of the portfolio loss using the data set $\tR^{[s]}$.  Hence, the Basel-TTCE risk measure can be estimated by
%\bea\label{equ:basel_ttce_sta}
%  & & \rho_{\text{Basel-TTCE}}(x(u))\\
%  &:=&\max\left\{\rho_{\TTCE}(x^{[1]}(u)),
%  \frac{k}{m_1}\sum_{s=1}^{m_1}\rho_{\TTCE}(x^{[s]}(u))\right\}\nonumber \\
%  &&+\max\left\{\rho_{\TTCE}(x^{[m_1+1]}(u)),
%  \frac{\ell}{m_2}\sum_{s=m_1+1}^{m}\rho_{\TTCE}(x^{[s]}(u))\right\},\nonumber
%\eea
%where $m_1=m_2=60$ and $m=120$.

\end{description}

\subsection{The ``Mean-$\rho$-Basel" Asset Allocation Model}

Suppose a portfolio manager in a financial institution attempts to construct a portfolio composed of the $d$
assets and to choose the portfolio weights $u\in \mathbb{R}^{d}$ to
optimize the portfolio performance. The manager can freely choose a risk measure $\rho$ to measure the risk of the portfolio, such as variance, VaR, or CVaR; in addition, he or she has the freedom to choose a model for the asset returns $R$ or a data set $\tilde Y\in \mathbb{R}^{n'\times d}$, which has a similar structure to that of $\tR$ defined in \eqref{equ:port_loss_obs} and contains observations of the asset returns, to estimate the portfolio risk. Hence, the portfolio risk will be given by $\rho(y(u))$, where $y(u):=-\tilde Yu$. Furthermore, the manager can specify that the expected portfolio return should be no less than a target return $r_0$, namely, the portfolio weights $u$ should satisfy
\[u \in \Ur :=\{ u \in \R^d \mid \mu^\top u \geq r_0, \one^\top u = 1, u \ge 0 \}. \]
Here, it is assumed that the portfolio is long only; this assumption can be relaxed or removed without incurring additional technical difficulty in solving the asset allocation problem specified below.

%The investor aims to minimize a risk measure subject to the expected return of the portfolio is not less than  $r_0$.

At the same time, the manager has to meet the constraint that the regulatory capital for his or her portfolio should not exceed an upper limit $C_0$, which is allocated to him or her by the financial institution's senior management. The capital requirement for the portfolio is calculated by the Basel Accord risk measure $\rho_{\BaselG}$, which is specified by the regulators; in addition, the data set $\tR$ used for calculating the capital requirements should also satisfy certain criteria and cannot be freely chosen by the portfolio manager. For example, the Basel 2.5 risk measure requires that $\tR$ should include 60 normal scenarios and 60 stressed scenarios. Hence, the data set $\tR$ may be different from the data set $\tilde Y$, and the capital requirement for the portfolio is $\rho_{\BaselG}(x(u))$, where $x(u)=-\tR u$.
%In order to address the constraint of capital requirement and control other
%risks, both internal  and external risk measures are required.
%The internal risk measure informs the investor/trader of the risk level of a
%portfolio; the external risk measure calculates the regulatory capital.
%Although
%a trader cannot choose the external risk measure of regulatory capital, which is
%imposed by regulators, he or she has the full freedom to choose the internal
%risk measure, such as the classical variance, VaR, and CVaR. In addition, the
%model or the data which he or she uses for the calculation of the internal risk
%measurement can be different from that used for the calculation of external risk
%measure.

To address the concerns of the portfolio manager, we propose the following \emph{``mean-$\rho$-Basel"} asset allocation problem:
\be\label{equ:opt_int_ext}
 \begin{aligned}
\min_{u \in \Ur }\;  & \rho(y(u))\\
\text{s.t. }\;\; & \rho_{\BaselG}(x(u)) \leq C_0,\end{aligned} \ee
where $x(u)=-\tR u$;
$y(u)=-\tilde Yu$;  $\rho_{\BaselG}$ is the Basel Accord risk measure for
calculating regulatory capital, i.e., $\rho_{\BaselTwoDotFive}$ or
$\rho_{\BaselThree}$; $C_0$ is the upper bound of the available capital; and $\rho$
is the risk measure that the manager chooses for gauging the risk
of the portfolio, such as variance, VaR, or CVaR.

%where $1_{n'}:=(1, 1, \ldots, 1)^\top\in \mathbb{R}^{n'}$.
%\be \label{eq:rho-var} \rho(y(u))=u^\top \hat \Sigma u, \ee
%where $\hat \Sigma$ is
%the covariance matrix of asset returns estimated from the data $\tilde Y$. For
%example, $\hat\Sigma$ can be estimated by the sample covariance matrix
%$\frac{1}{n'-1}(\tilde Y^\top \tilde Y-n'\bar Y\bar Y ^\top)$, where $\bar
%Y:=\frac{1}{n'}\tilde Y^\top 1_{n'}$ and $1_{n'}:=(1, 1, \ldots, 1)^\top\in \mathbb{R}^{n'}$.

%$y(u):=-\tilde Yu$ where $\tilde Y$ is a data set which may be different from $\tilde R$ used for defining $x(u)$;
%$\mu$ and $\Ur$ have the same definition as before.
%When $\rho$ is chosen to be variance, VaR, and CVaR respectively, then problem \eqref{equ:opt_int_ext} becomes the problem of mean-variance, mean-VaR, and mean-CVaR with Basel Accord capital constraint, respectively.
%Let $r_0$ be a given targeted expected return.  The investor aims to minimize a risk measure subject to the expected return of the portfolio is not less than  $r_0$.

%The problem \eqref{equ:opt_int_ext} is

The mean-$\rho$-Basel problem \eqref{equ:opt_int_ext}
with $\rho=\rho_{\VaR}$ or $\rho_{\BaselG}=\rho_{\BaselTwoDotFive}$ is non-convex and is usually difficult to solve, as it can be formulated as a mixed-integer programming (MIP) problem. For example, by introducing $z' \in \{0,1\}^{n'}$ and $z^{[s]} \in \{0,1\}^{n_s}$ for $1\leq s\leq m$, the mean-VaR-Basel2.5 problem can be
formulated as the following MIP problem:
\be\label{eqmip}
\begin{aligned}
\min_{u, z, \beta, \gamma} \ \ & \beta_0 \\
\text{s.t.}\ \ \  & -\tilde Yu \le \beta_0 \one +\eta z', \one^\top z' \le n'-p', \; z' \in \{0,1\}^{n'},\\
& -\tilde{R}^{[s]}u \le \beta_s \one +\eta z^{[s]}, \one^\top z^{[s]} \le n_s-p_s, z^{[s]} \in \{0,1\}^{n_s},\; s=1,\ldots,m,\\
%& \one^\top z^{[s]} \le n_s-p_s,\; z^{[s]} \in \{0,1\}^{n_s}, \; s=1,\ldots,m,\\
& \beta_1 \le \gamma_1,\;  \frac{k}{m_1}\sum_{s=1}^{m_1} \beta_s \le \gamma_1,
 \beta_{m_1+1} \le \gamma_2,\;  \frac{\ell}{m_2}\sum_{s=m_1+1}^{m} \beta_s \le \gamma_2,\\
&   \gamma_1 + \gamma_2 \le C_0,\\
& u \in \Ur,
\end{aligned}
\ee
where $p':=\lceil\alpha n'\rceil$, $p_s:=\lceil\alpha n_s\rceil$, $\eta$ is a large constant. For instance, $\eta$ can be chosen to be
$\eta=\max_{ u \in \Ur}  \max_{j=1,\ldots,n} (-\tilde Y u)_j$.

Similarly, by \eqref{equ:cvar_opt_rep}, the mean-VaR-Basel3 problem can be formulated as the following MIP problem:
\be\label{eqcmip}
\begin{aligned}
\min_{u, z', \beta_0, t, r} \ \ &  \beta_0  \\
\text{s.t.}\ \ \ \ \ \  & -\tilde Yu \le \beta_0 \one +\eta z', \one^\top z' \le n'-p',\; z' \in \{0,1\}^{n'},\\
& t_{m_1+1} + \frac{1}{(1-\alpha_3)n_{m_1+1}}\sum_{i=1}^{n_{m_1+1}}r_i^{[m_1+1]} \le C_0,\\
& \frac{\ell}{m_2}\sum_{s=m_1+1}^{m}(t_s + \frac{1}{(1-\alpha_3)n_s}\sum_{i=1}^{n_s}r_i^{[s]}) \le C_0,\\
&  r_i^{[s]} \ge 0,~ r_i^{[s]} \ge -\tilde R_i^{[s]}u-t_s,i=1,\ldots,n_s, s=m_1+1,\ldots,m,\\
& u \in \Ur.
\end{aligned}
\ee
On the other hand, the mean-$\rho$-Basel problem with $\rho_{\BaselG}$ being $\rho_{\BaselThree}$ and with $\rho$ being $\rho_{\Variance}$ or $\rho_{\CVaR}$ is convex. More precisely, the mean-variance-Basel3 and mean-CVaR-Basel3 problems
can be formulated as a quadratic programming (QP) problem  and a linear programming (LP) problem, respectively, thanks to the LP formulation of CVaR given in \eqref{equ:cvar_opt_rep}.

We develop a unified method for solving the mean-$\rho$-Basel problem
in Section \ref{subsubsec:algo_var} and provide convergence analysis of the method in Section \ref{sec:convergence}. The method can also be applied to solve the classical \emph{``mean-$\rho$"} problem:
\be\label{equ:opt} \min_ {u \in \Ur }  \;  \rho(x(u)), \ee
where $\rho$ can be any risk measure chosen by the portfolio manager, such as variance, VaR, CVaR, and $\rho_{\BaselG}$.
%The ``mean-$\rho$" problem does not take into account the Basel Accord capital constraint and is hence simpler than the problem \eqref{equ:opt_int_ext}.
%a simplification of
% the asset allocation problem  \eqref{equ:opt_int_ext} is
If $\rho=\rho_{\BaselG}$, problem \eqref{equ:opt} is the \emph{``mean-Basel"} problem, in which the Basel Accord risk measures are used to quantify the risk of the portfolio. The Basel Accord risk measures involve multiple VaR or CVaR under different scenarios, which essentially correspond to different models or distributions of  asset returns. Hence, using the Basel Accord risk measures, or, more generally, VaR or CVaR with scenario analysis, as the portfolio risk measure provides a way to address the problem of model uncertainty. Alternatively, the portfolio manager can construct the portfolio by maximizing the expected return of portfolio subject to the constraint that the portfolio risk, measured by $\rho$, does not exceed a pre-specified risk budget $b_0$. The corresponding asset allocation problem is
\be\label{equ:optc} \begin{aligned} \min_{u \in \U } & \; -\mu^\top u \\
  \text{s.t.} &\;  \rho(x(u)) \leq b_0, \end{aligned} \ee
where
$ \U =\{ u \in \R^d \mid  \one^\top u = 1, u \ge 0 \}$.
The mean-$\rho$ problems
\eqref{equ:opt} and \eqref{equ:optc} with
$\rho \in \{\rho_{\VaR}, \rho_{\BaselTwoDotFive}\}$ are also
MIP problems which are difficult to solve. The details of the method for solving these problems are given in Section \ref{sec:ADM-other-models}.
%In fact,
%the problem \eqref{equ:opt} with
%$\rho=\rho_{\VaR}$ can be
%formulated as
%\be\label{equ:opt_mip}
%\begin{aligned}
%  \min_{u \in \Ur, \gamma \in \R, y \in \R^n}\ \ & \gamma
%\\
%\text{s.t. } \; \;\qquad &  - \tilde R u \le \gamma +\eta y, \\
%& \one^\top y = n-p, y\in \{0,1\}^n,
%\end{aligned}
%\ee
%where $\eta$ is a large constant, for example,
% $ \eta=\max_{ u \in \Ur}  \max_{j=1,\ldots,n} (-\tR u)_j$.
%Similarly, problem \eqref{equ:opt} with $\rho=\rho_{\BaselTwoDotFive}$ can be rewritten as:
%\be\label{eqmip}
%\begin{aligned}
%  \min \ \ & \gamma_1 + \gamma_2 \\
%  \text{s.t. }\; \;  & -\tilde{R}^{[s]}u \le \beta_s \one +\eta z^{[s]}, \; s=1,\ldots,m,\\
% & \one^\top z^{[s]} \le n_s-p_s,\; z^{[s]} \in \{0,1\}^{n_s}, \; s=1,\ldots,m,\\
%  & \beta_1 \le \gamma_1,\;  \frac{k}{m_1}\sum_{s=1}^{m_1} \beta_s \le \gamma_1,\\
%  & \beta_{m_1+1} \le \gamma_2,\;  \frac{\ell}{m_2}\sum_{s=m_1+1}^{m} \beta_s \le \gamma_2,
%   u \in \Ur.
%  \end{aligned}
%\ee
%
%On the other hand, by using a similar proof as Theorem 16 in \citet{Rock02},
% problem \eqref{equ:opt} with $\rho=\rho_{\BaselThree}$  can be
%formulated as a linear programming problem:
%\be\label{eqcvar}
%\begin{aligned}
%  \min \;\; & \gamma \\
%\text{s.t. }\; \;  & t_{m_1+1} + \frac{1}{(1-\alpha)n_{m_1+1}}\sum_{i=1}^{n_{m_1+1}}y_i^{[m_1+1]} \le \gamma,\\
% & \frac{\ell}{m_2}\sum_{s=m_1+1}^{m}(t_s + \frac{1}{(1-\alpha)n_s}\sum_{i=1}^{n_s}y_i^{[s]}) \le \gamma,\\
%&  y_i^{[s]} \ge 0,~ y_i^{[s]} \ge -\tilde{R}_i^{[s]}u-t_s, \; i=1, \ldots, n_s; s=m_1+1,\ldots,m,  u \in \Ur.
%  \end{aligned}
%\ee

\section{The Alternating Direction Augmented Lagrangian Method}\label{sec:adm_algo}
%For simplicity of presentation, we focus on the model problem \eqref{equ:opt} in
%this section. Similar algorithms and analysis can be extended to models \eqref{equ:optc} and
%\eqref{equ:opt_int_ext} without any difficulty.
% Specifically, we first present the ADM framework for solving \eqref{equ:opt} in
%subsection \ref{subsubsec:algo_var}, then demonstrate how the algorithm can
%be used to solve other models in subsection  \ref{sec:ADM-other-models}, and
%finally provide convergence analysis for  \eqref{equ:opt} in subsection
% \ref{sec:convergence}.
In this section, we propose a unified algorithm adapted from the
 \emph{alternating direction augmented Lagrangian} method (ADM) to solve the mean-$\rho$-Basel and the mean-$\rho$ problem. Although the ADM approach has been
used in convex optimization (see, e.g., \citealt*{WGY-2010}; \citealt{HeYuan2012}; and \citealt*{HongLuo2012}),
it appears that its use for
solving non-convex problems involving VaR or Basel Accords risk measures is new. In particular, the proposed method is different from the penalty
decomposition methods proposed in \citet{BaiZhengSunSun2012}, in which the division of blocks of variables leads to subproblems that
are more expensive to solve.  %\revise{For simplicity of presentation, we first present the method for solving the mean-$\rho$ problem \eqref{equ:opt} in Section \ref{subsubsec:algo_var} and provide convergence analysis for the method in Section \ref{sec:convergence}. We then demonstrate the method for solving the mean-$\rho$-Basel model \eqref{equ:opt_int_ext} in Section \ref{sec:ADM-other-models}.}

\subsection{The ADM Algorithm for Solving the Mean-$\rho$-Basel Problem \eqref{equ:opt_int_ext}}\label{subsubsec:algo_var}
% By introducing certain intermediate variables, we can exploit the structures of
% the risk measures and design an  efficient approximate algorithm.
% Treating
% the implicit relationships $x(u)$ and $y(u)$ as variables $x$ and $y$,
% respectively, and imposing $x = -\tR u$ and $y = -\tilde Y u $
% explicitly, we can rewrite \eqref{equ:opt_int_ext} as
The problem \eqref{equ:opt_int_ext} is equivalent to
\be\label{equ:opt_int_ext2}
 \begin{aligned}
\min_{u \in \Ur, x\in\mathbb{R}^n, y\in\mathbb{R}^{n'}}   & \; \rho(y)\\
\text{s.t.}\ \ \ \ \ \ \quad &\; \rho_{\BaselG}(x) \leq C_0, \\
& x + \tR u = 0, \\
& y + \tilde Y u = 0.
\end{aligned} \ee
We then define the augmented
Lagrangian function for \eqref{equ:opt_int_ext2} as follows:
\be
\Lc(x,y,u,\lambda,\pi):=\rho(y) + \lambda^\top (x + \tR u) + \frac{\sigma_1}{2}
\| x + \tR u \|^2 + \pi^\top (y + \tilde Y u) + \frac{\sigma_2}{2} \|y + \tilde Y u  \|^2,
\label{auglang}
\ee
where $\sigma_1,\sigma_2>0$ is the penalty parameter and $\lambda \in \R^n$ and
$\pi \in \R^{n'}$  are the Lagrangian multipliers associated with
the equality constraints $x+\tR u=0$ and $y + \tilde Y u = 0$, respectively.

%\begin{algorithm2e}[H]
%%\SetAlgoLined
%\KwData{this text}
%\KwResult{how to write algorithm with \LaTeX2e }
%initialization\;
%\While{not at end of this document}{
%read current\;
%\eIf{understand}{
%go to next section\;
%current section becomes this one\;
%}{
%go back to the beginning of current section\;
%}
%}
%\caption{How to write algorithms}
%\end{algorithm2e}

We propose an ADM algorithm that minimizes \eqref{auglang} with respect to $x$, $y$, and $u$ in an alternating fashion while updating $\lambda$ and $\pi$ in the iteration. More precisely, let $x^{(j)}, y^{(j)}$, and $u^{(j)}$ be the values of $x$, $y$, and $u$ at the beginning of the $j$th iteration of the algorithm; then the algorithm updates the values of $x$, $y$, and $u$ by solving the following three subproblems sequentially:
\begin{align}
x^{(j+1)}  =& \arg \min_{x \in \R^n}  \;
\Lc(x, y^{(j)}, u^{(j)}, \lambda^{(j)}, \pi^{(j)}), \text{ s.t. }
\rho_{\BaselG}(x) \le C_0,\label{eq:ADM-x}\\
y^{(j+1)}  =& \arg \min_{y \in \R^{n'}}  \;
\Lc(x^{(j+1)}, y, u^{(j)}, \lambda^{(j)}, \pi^{(j)}),\label{eq:ADM-y}\\
 u^{(j+1)}  =&  \arg \min_{u \in \Ur } \;  \Lc(x^{(j+1)},y^{(j+1)}, u, \lambda^{(j)}, \pi^{(j)}).\label{eq:ADM-u}
\end{align}
Then, it updates the the Lagrangian multipliers by
\begin{equation}
 \lambda^{(j+1)} = \lambda^{(j)} + \beta_1 \sigma_1 (x^{(j+1)}+ \tR u^{(j+1)}), \label{eq:ADM-lmb}
\end{equation}
%\end{algomathdisplay}
\begin{equation}\label{eq:ADM-pi}
 \pi^{(j+1)} =\pi^{(j)} + \beta_2 \sigma_2 (y^{(j+1)}+ \tilde Y u^{(j+1)}),
\end{equation}
where $\beta_1,\beta_2>0$ are appropriately chosen step lengths.

The solutions to problems \eqref{eq:ADM-x} and \eqref{eq:ADM-y} are given in the lemmas at the end of this subsection; these solutions are obtained either in closed form, by solving QP problems, or by minimizing a single variable function on a closed interval. As for
problem \eqref{eq:ADM-u}, simple algebra shows that it is equivalent to the following QP problem \eqref{eq:ADM-u-v2}:
\begin{align} \label{eq:ADM-u-v2} %u^{(j+1)}  =   \arg  \min_{ u \in \Ur} \;   \frac{1}{2} u^\top \tR^\top \tR u + b^\top u,
 u^{(j+1)}  =  & \arg \min_{u \in \Ur } \;   \frac{1}{2} u^\top (\sigma_1
 \tR^\top \tR  + \sigma_2 \tilde Y ^\top \tilde Y )  u +
 b_e^\top u,\ \text{where}\\
%$b=\tR^\top (\frac{1}{\sigma} \lambda^{(j)}  + x^{(j+1)})$.
& b_e=\tR^\top ( \lambda^{(j)} + \sigma_1 x^{(j+1)}) + \tilde Y^\top ( \pi^{(j)}
+ \sigma_2 y^{(j+1)}).\notag
\end{align}

The complete ADM algorithm is given as follows.

\begin{algorithm}
\caption{ADM algorithm for solving the mean-$\rho$-Basel problem \eqref{equ:opt_int_ext}}
\begin{algorithmic}[1]
\STATE Choose parameter $\sigma_1>0$, $\sigma_2>0$, $\beta_1>0$, $\beta_2>0$\;
\STATE Set $j=0$; initialize $y^{(0)} \in \R^{n'}$, $u^{(0)} \in \R^d$,
$\lambda^{(0)}:=0$, and $\pi^{(0)}:=0$\;
\WHILE{$\{u^{(j)}\}$ has not converged}
\STATE update $x^{(j+1)}$ to be the solution to  problem \eqref{eq:ADM-x}; the solution is given in Lemma \ref{lemma:opt-ADM-basel-x} and Lemma \ref{lemma:opt-ADM-basel3-x}
for $\rho_{\BaselG}=\rho_{\BaselTwoDotFive}$ and $\rho_{\BaselG}=\rho_{\BaselThree}$, respectively\;
\STATE update $y^{(j+1)}$ to be the solution to  problem \eqref{eq:ADM-y}; the solution is given in Lemma \ref{lemma:opt-ADM-y-variance}, Lemma \ref{lemma:opt-ADM-y-var}, and Lemma \ref{lemma:opt-ADM-y-cvar} for $\rho=\rho_{\Variance}$, $\rho=\rho_{\VaR}$, and $\rho=\rho_{\CVaR}$, respectively\;
\STATE update $u^{(j+1)}$ by solving the QP problem \eqref{eq:ADM-u-v2}\;
\STATE update $\lambda^{(j+1)}$ and $\pi^{(j+1)}$ by \eqref{eq:ADM-lmb} and \eqref{eq:ADM-pi}, respectively\;
\STATE increase $j$ by one and continue.
\ENDWHILE
\end{algorithmic}
\end{algorithm}

The algorithm is very simple and easy to implement. Standard QP solvers, such as CPLEX, can be used to solve the QP problems in Step 4, 5, and 6 of the algorithm. In step 5 for the case of $\rho=\rho_{\CVaR}$, the solution is obtained by minimizing a single-variable
function on a closed interval, which can be solved by golden section search and parabolic
interpolation (e.g., the function ``fminbnd'' in Matlab).

One particular implementation of the ADM algorithm including the specification of the parameters $\sigma_i$ and $\beta_i$ and the convergence test is given in Section \ref{subsec:parameters_adm}.

The lemmas for solving the subproblems in the algorithm are as follows.

\begin{lemma} \label{lemma:opt-ADM-basel-x}
Consider problem \eqref{eq:ADM-x} with $\rho_{\BaselG}=\rho_{\BaselTwoDotFive}$.
Let $v:=-\left(\tR u^{(j)} + \frac{1}{\sigma_1} \lambda^{(j)} \right)$
and denote $v=((v^{[1]})^{\top}, (v^{[2]})^{\top}, \ldots, (v^{[m]})^{\top})^{\top}$, where $v^{[s]}\in\mathbb{R}^{n_s}$, $s=1, 2, \ldots, m$. Let $(k_{s,1}, k_{s, 2}, \ldots, k_{s, n_s})$ be the permutation of $(1, 2, \ldots, n_s)$ such that $v^{[s]}_{k_{s, 1}}\leq v^{[s]}_{k_{s, 2}}\leq \cdots \leq v^{[s]}_{k_{s, n_s}}$, $s=1,\ldots,m$. Let $p_s:=\lceil\alpha n_s\rceil$ and $h^{[s]}:=(v^{[s]}_{k_{s, 1}}, v^{[s]}_{k_{s, 2}}, \ldots, v^{[s]}_{k_{s, p_s}})^{\top}$, $s=1,\ldots,m$. The optimal solution $x$ to \eqref{eq:ADM-x} is given by
 \begin{eqnarray}\label{eq:solution-2}
x^{[s]}_{k_{s, i}} = \begin{cases} z^{[s]}_{i}, &\mbox{if}\
1\leq i\leq p_s,\\ v^{[s]}_{k_{s, i}}, &\mbox{otherwise},
\end{cases}\ i=1, 2, \ldots, n_s,
\end{eqnarray}
where $(z^{[1]}, z^{[2]}, \ldots, z^{[m]})$ is the optimal solution to the
following QP problem:
\be\label{eq:QP}
\begin{aligned}
\min\limits_{z, \gamma_1, \gamma_2}\; \; \;&   \sum_{s=1}^{m} \| z^{[s]} - h^{[s]}\|^2 \\
%\mbox{\rm s.t.}\;\quad& z^{[s]} \leq z_{p_s}^{[s]} \one,\quad s=1,\cdots,m,\\
\text{s.t.}\;\quad& z_1^{[s]} \leq z_2^{[s]}\leq \cdots \leq z_{p_s}^{[s]},\quad s=1,\cdots,m,\\
& \gamma_1 + \gamma_2 \le C_0, \quad z_{p_1}^{[1]} \leq \gamma_1,  \quad z_{p_{m_1+1}}^{[m_1+1]} \leq \gamma_2, \\
&\frac{k}{m_1}\sum_{s=1}^{m_1}z_{p_s}^{[s]} \leq \gamma_1, \quad \frac{\ell}{m_2}\sum_{s=m_1+1}^{m}z_{p_s}^{[s]} \leq \gamma_2.
\end{aligned}
\ee
\end{lemma}
%\begin{lemma} \label{lemma:opt-ADM-basel-x_o} Suppose, without loss of generality, that $v^{[s]}_1\leq v^{[s]}_2\leq \cdots \leq v^{[s]}_{n_s}$, $s=1,\ldots,m$.
% The optimal solution $x$ of \eqref{eq:ADM-x-v5} are given by
% \begin{eqnarray}\label{eq:solution-2_o}
%x^{[s]}_{i} = \left\{\begin{array}{ll} z^{[s]}_{i}, &\mbox{if}\
%1\leq i\leq p_s,\\ v^{[s]}_{i}, &\mbox{otherwise},
%\end{array}\right.\ i=1, 2, \ldots, n_s,
%\end{eqnarray}
%where $z$ in  \eqref{eq:solution-2_o} are determined by the optimal solutions of the
%following QP:
%\be\label{eq:QP_o}
%\begin{aligned}
%\min\limits_{z, \gamma_1, \gamma_2}\; \; \;&   \sum_{s=1}^{m} \| z^{[s]} - h^{[s]}\|^2 \\
%\mbox{\rm s.t.}\;\quad& z^{[s]} \leq z_{p_s}^{[s]} \one,\quad s=1,\cdots,m,\\
%& \gamma_1 + \gamma_2 \le C_0, \quad z_{p_1}^{[1]} \leq \gamma_1,  \quad z_{p_{m_1+1}}^{[m_1+1]} \leq \gamma_2, \\
%&\frac{k}{m_1}\sum_{s=1}^{m_1}z_{p_s}^{[s]} \leq \gamma_1, \quad \frac{\ell}{m_2}\sum_{s=m_1+1}^{m}z_{p_s}^{[s]} \leq \gamma_2.
%\end{aligned}
%\ee
%Here  $h^{[s]}$ are vectors consisting of the first $p_s$ smallest components of
%$v^{[s]}$, $s=1,\ldots,m$.
%\end{lemma}
\begin{proof} See Appendix \ref{app:proof-lemma-opt-ADM-basel-x}.
\end{proof}

\begin{lemma} \label{lemma:opt-ADM-basel3-x}
Consider problem \eqref{eq:ADM-x} with
$\rho_{\BaselG}=\rho_{\BaselThree}$.
Let $v$ and $v^{[s]}$ be defined as in Lemma \ref{lemma:opt-ADM-basel-x}. Let $x^{[s]}$,
  $s=m_1+1,\ldots,m$ be the optimal solution to the following QP problem:
%\be\label{eq:ADM-x-CVaR-solu}
%\begin{aligned}
%  \min_{t,x} \;\; & \sum_{s=m_1+1}^m \| x^{[s]} - v^{[s]}\|^2, \\
%\text{s.t. }\; \;  & t_{m_1+1} +
%\frac{1}{(1-\alpha)n_{m_1+1}}\sum_{i=1}^{n_{m_1+1}}(x_i^{[m_1+1]} -t_{m_1+1})_+  \le C_0,\\
% & \frac{\ell}{m_2}\sum_{s=m_1+1}^{m}(t_s +
% \frac{1}{(1-\alpha)n_s}\sum_{i=1}^{n_s}( x_i^{[s]} - t_s)_+) \le C_0.   \\
%  \end{aligned}
%\ee
\be\label{eq:ADM-x-CVaR-solu}
\begin{aligned}
  \min_{t,x,z} \;\; & \sum_{s=m_1+1}^m \| x^{[s]} - v^{[s]}\|^2, \\
\text{s.t. }\; \;  & t_{m_1+1} +
\frac{1}{(1-\alpha)n_{m_1+1}}\sum_{i=1}^{n_{m_1+1}}z_i^{[m_1+1]}  \le C_0\\
 & \frac{\ell}{m_2}\sum_{s=m_1+1}^{m}(t_s +
 \frac{1}{(1-\alpha)n_s}\sum_{i=1}^{n_s}z_i^{[s]}) \le C_0,\\
& z_i^{[s]}\geq 0, z_i^{[s]}\geq x_i^{[s]} -t_{s}, i=1, \ldots, n_s, s=m_1+1, \ldots, m.\\
\end{aligned}
\ee
Then the optimal solution to \eqref{eq:ADM-x} is given by
$$x=((v^{[1]})^{\top}, (v^{[2]})^{\top}, \ldots, (v^{[m_1]})^{\top}, (x^{[m_1+1]})^{\top}, (x^{[m_1+2]})^{\top}, \ldots, (x^{[m]})^{\top})^{\top}.$$
\end{lemma}
\begin{proof} See Appendix \ref{app:proof-lemma-opt-ADM-basel3-x}.
\end{proof}
%The interval $[\min_{1\le i \le n'} w_i -c, \max_{1\le i \le n'} w_i]$ is divided by the points $\{w_i\}$ into $n'$ subintervals. For $t$ in any given subinterval,
% $y_i(t)$ is a function of $t$ as \eqref{equ:xh-CVaR-yt}. Substituting
% them into $\phi(t,y(t))$ gives a polynomial of $t$. The optimal solution
% of \eqref{equ:xh-CVaR-t} is found if the corresponding
% minimum of this polynomial lies in this interval. Hence, at most $n'$ steps of
% such comparison is needed and the computational complexity is $O( (n')^2)$.
% Practically, algorithms based on golden section search and parabolic
% interpolation, for example, the function ``fminbnd'' in Matlab, can be used.
%\item[$y$-subproblem \eqref{eq:ADM-y}]
%\subsubsection{Solving the subproblem \eqref{eq:ADM-y}.}
%\be \label{eq:ADM-x-v2}
%\begin{aligned} x^{(j+1)}  =   \arg \min_{x}  \; &
%  \psi(x)=
%  \rho(x) + \frac{\sigma}{2} \| x - v^{(j)}\|^2, \end{aligned}
% \ee
% where $v^{(j)} = - \left(\tR u^{(j)} + \frac{1}{\sigma} \lambda^{(j)} \right)$.
 %The optimal solutions of the above problem will be explained later.

%whose solution $(t,y)$ satisfies
%\be \label{equ:xh-CVaR-y} y_i = \begin{cases} w_i -c, & \mbox{ if } y_i >
%  t, \\
%w_i -\gamma_i c, & \mbox{ if } y_i =
%  t, \\
%  w_i, & \mbox{ if } y_i<
%  t,
%\end{cases}\ee
%where $c =\frac{1}{\sigma_2 (1-\alpha)n'} $ and $\gamma_i \in [0,1]$.

\begin{lemma}\label{lemma:opt-ADM-y-variance}
The optimal solution to problem \eqref{eq:ADM-y} with $\rho=\rho_{\Variance}$ is
%\be\label{equ:adm-y-variance}
%y^{(j+1)}=-\sigma_2
%\left((\sigma_2+\frac{2}{n'})I -\frac{2}{(n')^2}\one\one^{\top}\right)^{-1} \left(\tilde Y u^{(j)} + \frac{1}{\sigma_2} \pi^{(j)} \right).
%\ee
\be\label{equ:adm-y-variance}
y^{(j+1)}=\sigma_2
\left((\sigma_2+\frac{2}{n'})I
-\frac{2}{(n')^2}\one\one^{\top}\right)^{-1}w^{(j)}= \left(\sigma_2+\frac{2}{n'}\right)^{-1} \left( \sigma_2 w^{(j)} + 2\frac{\one^{\top} w^{(j)}}{(n')^2}
 \one  \right),
\ee
where $w^{(j)} = - \left(\tilde Y u^{(j)} + \frac{1}{\sigma_2} \pi^{(j)}
\right)$.
\end{lemma}
%\begin{lemma}\label{lemma:opt-ADM-y-var} %Consider \eqref{eq:ADM-y-v4} with $\rho(y(u))=\rho_{\VaR}(y(u))$.
% Suppose, without loss of generality, that $ w_1 \le w_2 \le \ldots \le w_{n'}$.
% The optimal solution $y$ of \eqref{equ:xh_1} is given by \begin{eqnarray}\label{eq:solution}
%y_{(i)} = \left\{\begin{array}{ccc} \gamma_{i^*}, &\mbox{if}&
% i^* \leq i\leq p';\\ w_i, && \mbox{otherwise},
%\end{array}\right.
%\end{eqnarray}
%where
%\be \label{equ:opt_x_order-v2}
%i^*:= \max \{ i \mid i\le p', \; w_{i-1} < \gamma_i  \}  \mbox{ and } \gamma_i
%= \frac{ \sigma_2 \sum_{j=i}^{p'} w_j - 1}{ \sigma_2 (p'-i+1) }.
%\ee
%\end{lemma}
\begin{proof} See Appendix \ref{sec:app_lemma_opt-ADM-y-variance}.
\end{proof}

\begin{lemma}\label{lemma:opt-ADM-y-var}
Consider problem \eqref{eq:ADM-y} with $\rho=\rho_{\VaR}$. Define
$w: = - (\tilde Y u^{(j)} + \frac{1}{\sigma_2} \pi^{(j)})\in\mathbb{R}^{n'}$ and $p':=\lceil \alpha n'\rceil$. Let $(k_1, k_2, \ldots, k_{n'})$ be the permutation of $(1, 2, \ldots, n')$ such that $ w_{k_1} \le w_{k_2} \le \cdots \le w_{k_{n'}}$.
Then the optimal solution $y$ to \eqref{eq:ADM-y} with $\rho=\rho_{\VaR}$ is given by \begin{eqnarray}\label{eq:solution}
y_{k_i} = \begin{cases} \gamma_{i^*}, &\mbox{if}\
 i^* \leq i\leq p',\\ w_{k_i}, & \mbox{otherwise},
\end{cases}\ i=1, 2, \ldots, n',
\end{eqnarray}
%\be \label{equ:opt_x_order-v2}
%i^*:= \max \{ i \mid i\le p', \; w_{k_{i-1}} < \gamma_i  \},\ w_{k_0}:=-\infty,  \mbox{ and } \gamma_i
%:= \frac{ \sigma_2 \sum_{j=i}^{p'} w_{k_j} - 1}{ \sigma_2 (p'-i+1) }.
%\ee
where
\begin{equation} \label{equ:opt_x_order-v2}
i^*:= \max \{ i \mid 1\leq i\le p', \; w_{k_{i-1}} < \gamma_i\leq w_{k_{i}}  \},\ w_{k_0}:=-\infty,  \mbox{ and } \gamma_i
:= \frac{ \sigma_2 \sum_{j=i}^{p'} w_{k_j} - 1}{ \sigma_2 (p'-i+1) }.
\end{equation}
\end{lemma}
%\begin{lemma}\label{lemma:opt-ADM-y-var} %Consider \eqref{eq:ADM-y-v4} with $\rho(y(u))=\rho_{\VaR}(y(u))$.
% Suppose, without loss of generality, that $ w_1 \le w_2 \le \ldots \le w_{n'}$.
% The optimal solution $y$ of \eqref{equ:xh_1} is given by \begin{eqnarray}\label{eq:solution}
%y_{(i)} = \left\{\begin{array}{ccc} \gamma_{i^*}, &\mbox{if}&
% i^* \leq i\leq p';\\ w_i, && \mbox{otherwise},
%\end{array}\right.
%\end{eqnarray}
%where
%\be \label{equ:opt_x_order-v2}
%i^*:= \max \{ i \mid i\le p', \; w_{i-1} < \gamma_i  \}  \mbox{ and } \gamma_i
%= \frac{ \sigma_2 \sum_{j=i}^{p'} w_j - 1}{ \sigma_2 (p'-i+1) }.
%\ee
%\end{lemma}
\begin{proof} See Appendix \ref{sec:app_lemma_opt-ADM-y-var}.
\end{proof}

\begin{lemma}\label{lemma:opt-ADM-y-cvar}
Consider problem \eqref{eq:ADM-y} with $\rho=\rho_{\CVaR}$. Let $w$ be defined as in Lemma \ref{lemma:opt-ADM-y-var}. Define
\be\label{equ:def_phi}
\phi(t,y):=
    t + \frac{1}{(1-\alpha)n'}\sum_{i=1}^{n'} (y_i-t)_+ + \frac{\sigma_2}{2} \| y - w\|^2,
\ee
where $x_+:=\max(x, 0)$. Let $t^*$ be the optimal solution to
  \be \label{equ:xh-CVaR-t} \min_t \; \phi(t, y(t)), \; \text{ s.t. } \;
  \min_{1\le i \le n'} w_i -c\le t \le
  \max_{1\le i \le n'} w_i,\ee
where $c= \frac{1}{\sigma_2 (1-\alpha)n'}$; $y(t)=(y_1(t), y_2(t), \ldots, y_{n'}(t))^{\top}$; $y_i(t), i=1,\ldots,n'$ are defined by
\be \label{equ:xh-CVaR-yt} y_i(t) = \begin{cases} w_i -c, & \mbox{ if }  w_i -c> t, \\
t, & \mbox{ if }  w_i > t \ge w_i - c, \\
  w_i, & \mbox{ otherwise. }
\end{cases}\ee
Then $y(t^*)$ is the optimal solution to \eqref{eq:ADM-y} with $\rho=\rho_{\CVaR}$.
\end{lemma}
\begin{proof} See Appendix \ref{app:proof-lemma-opt-ADM-y-cvar}.
\end{proof}

\subsection{Convergence Analysis} \label{sec:convergence}
If $\rho$ is $\rho_{\Variance}$ or $\rho_{\CVaR}$ and $\rho_{\BaselG}$ is $\rho_{\BaselThree}$, problem \eqref{equ:opt_int_ext} is convex and the
 ADM method is ensured to converge to the global solutions theoretically \citep{HongLuo2012}.
On the other hand, if $\rho$ is $\rho_{\VaR}$ or if $\rho_{\BaselG}$ is $\rho_{\BaselTwoDotFive}$,
%ince VaR and Basel 2.5 risk measure are non-convex functions,
the
convergence of the ADM algorithm to a global
optimal solution is not guaranteed due to the non-convexity of $\rho_{\VaR}$ or $\rho_{\BaselTwoDotFive}$; however, we will show in the following that
the limit point of the sequence generated by the ADM algorithm satisfies the first-order optimality conditions of problem \eqref{equ:opt_int_ext} under some mild conditions. In addition, numerical experiments suggest
that the ADM algorithm seems to
converge from any starting point.

%Although far from satisfactory, this result nevertheless provides some assurance on the reliability of the algorithm.
We first recall the definition of locally Lipschitz functions.
% and the Clarke's
%generalized gradient.
\begin{definition}\label{def:1} A function $f(x):\mathrm{dom}f\subseteq\R^n \to \R$ is Lipschitz near a point
  $x_0 \in \mathrm{int}(\mathrm{dom}f)$ if there exist $K\ge 0$ and $\delta>0$ such that $|f(x)-f(x')|\le
  K \|x-x'\|$ for all $x, x' \in B_\delta(x_0)$, where $B_\delta(x_0):=\{x\in\R^{n}:\,\|x-x_0\|<\delta\}\subseteq \mathrm{dom}f$.
  A function is locally Lipschitz if it is Lipschitz
  near every point in $\R^{n}$. A function is globally Lipschitz on $\mathbb{R}^n$ if there exists a constant $K\ge 0$ such that $|f(x)-f(x')|\le  K \|x-x'\|$ for all $x, x' \in \mathbb{R}^n$.
\end{definition}

We have the following result on the Lipschitz property of risk measures.

%, such as VaR and Basel
%2.5 risk measure, are locally Lipschitz. %\begin{definition}\label{def:3}
\begin{proposition}\label{lm:1}
The functions $\rho_{\VaR}(x)$, $\rho_{\CVaR}(x)$,  $\rho_{\BaselTwoDotFive}(x)$, and $\rho_{\BaselThree}(x)$ defined in Section \ref{subsec:risk_stat_def} are all globally Lipschitz on $\mathbb{R}^n$.
\end{proposition}
\begin{proof} See Appendix \ref{app:proof_lm1}.
\end{proof}

%We have the following first-order optimality (KKT) conditions of \eqref{equ:opt_int_ext}.
%\begin{proposition} Suppose that $\rho(x)$ is locally Lipschitz. %$x(u)=-\tR u$.
%If $u$ is a local minimizer of \eqref{equ:opt_int_ext}, then there exits $\eta \ge 0$, such that
%\bea\label{eq:kkt-30}
%0 \in  -\tY^\top \bar \partial \rho(y(u))- \eta\tR^\top \bar \partial \rho_{\BaselG}(x(u)) + \NU(u), \\
%\label{eq:kkt-31}\eta  (\rho_{\BaselG}(x(u)) - C_0) = 0  % was + and + before
%\eea
%where $\NU(u)={\rm
%cl}\left\{\cup_{t\geq0}t\,\partial\dist_{\Ur}(u)\right\}$ is the
%normal cone to $\Ur$ at $u$.% (see the detailed definition in subsection 2.4 in \citep{Clarke1990}).
%\end{proposition}

%Hence, Proposition \ref{ps:1} and Lemma \ref{lm:1} imply that the Clarke's generalized gradients of
%$\rho_{\VaR}(x)$ and $\rho_{\BaselTwoDotFive}(x)$
%exist over $\R^{n}$ and  their forms are described in Appendix \ref{CGG}.

We have the following theorem regarding the optimality of the output of the ADM algorithm.

\begin{theorem} \label{thm:cvg-feasi} %Suppose that $\sigma_1 \tR^\top \tR+\sigma_2 \tY^\top \tY\succ 0$.
Suppose that $\rho(x)$ is locally Lipschitz and $\rho_{\BaselG}\in\{\rho_{\BaselTwoDotFive}, \rho_{\BaselThree}\}$, then the following statements hold.

(i) (KKT conditions) If $u$ is a local minimizer of \eqref{equ:opt_int_ext}, then there exits $\eta \ge 0$, such that
\bea\label{eq:kkt-30}
0 \in  -\tY^\top \bar \partial \rho(y(u))- \eta\tR^\top \bar \partial \rho_{\BaselG}(x(u)) + \NU(u), \\
\label{eq:kkt-31}\eta  (\rho_{\BaselG}(x(u)) - C_0) = 0,  % was + and + before
\eea
where $\NU(u)$ is the
normal cone to $\Ur$ at $u$ and $\bar\partial f(\cdot)$ denotes the Clarke's generalized gradient of $f(\cdot)$.% (see the detailed definition in subsection 2.4 in \citep{Clarke1990}).

(ii) Let
  $\{(x^{(j)}, y^{(j)}, u^{(j)},\lambda^{(j)},\pi^{(j)})\}$ be a sequence generated by
  scheme \eqref{eq:ADM-x}-\eqref{eq:ADM-pi} and assume that
  $\sum_{j=1}^{\infty}   \|\lambda^{(j+1)} - \lambda^{(j)}\|^2 + \|\pi^{(j+1)} -
  \pi^{(j)}\|^2<\infty$ and $\{(\lambda^{(j)}, \pi^{(j)})\}$ is bounded.
  Then, the sequence $\{u^{(j)}\}$ is bounded and any limit point $\bar u$ of $\{u^{(j)}\}$ satisfies the first-order
  optimality  conditions \eqref{eq:kkt-30}-\eqref{eq:kkt-31}. \end{theorem}
\begin{proof} See Appendix \ref{app:proof-theorem-thm:cvg-feasi}.
\end{proof}

By Proposition \ref{lm:1}, Theorem \ref{thm:cvg-feasi} applies to the ADM algorithm with $\rho$ being $\rho_{\Variance}$, $\rho_{\VaR}$, or $\rho_{\CVaR}$.

%The proof of Theorem \ref{thm:cvg-feasi} is given in the appendix.

\subsection{The ADM Algorithm for Solving the Mean-$\rho$ Problems \eqref{equ:opt} and \eqref{equ:optc}} \label{sec:ADM-other-models}
%Let $x=-\tR u$.
The ADM algorithm for solving the mean-$\rho$ problems including the mean-VaR and mean-Basel  problems is as follows.

%In this subsection, we briefly show  how the ADM framework can be
%adapted to solve other two models.

\begin{description}

\item[\textbf{ADM for Solving Problem \eqref{equ:opt}:}]
The  augmented
Lagrangian function  for \eqref{equ:opt} is defined  as
\be
\Lc(x,u,\lambda):=  \rho(x) + \lambda^\top (x + \tR u) + \frac{\sigma}{2} \| x + \tR u \|^2,
\label{auglang-v3}
\ee
where $\sigma>0$ is the penalty parameter and $\lambda \in \R^n$ is the
Lagrangian multiplier. The  ADM method is
%In fact, we have:
\bea \label{eq:ADM-x3}
x^{(j+1)}  &=& \arg \min_{x \in \R^n}  \;
 \rho(x) + \frac{\sigma}{2} \| x - v^{(j)}\|^2, \\
\label{eq:ADM-u3}
 u^{(j+1)}  &=&  \arg \min_{u \in  \Ur} \; \frac{1}{2} u^\top \tR^\top \tR u + b^\top u ,
  \\
 \label{eq:ADM-lmb3}
\lambda^{(j+1)} &=&  \lambda^{(j)} + \beta \sigma (x^{(j+1)}+ \tR u^{(j+1)}),
\eea
where $v^{(j)} = - \left(\tR u^{(j)} + \frac{1}{\sigma} \lambda^{(j)} \right)$, $b=\tR^\top (\frac{1}{\sigma} \lambda^{(j)}  + x^{(j+1)})$,  and $\beta>0$. For $\rho=\rho_{\Variance}$, $\rho_{\VaR}$, and $\rho_{\CVaR}$, the subproblem \eqref{eq:ADM-x3} is the same as \eqref{eq:ADM-y} and hence its solution is given by Lemma \ref{lemma:opt-ADM-y-variance}, \ref{lemma:opt-ADM-y-var}, and \ref{lemma:opt-ADM-y-cvar}, respectively; for $\rho=\rho_{\BaselG}$, problem \eqref{eq:ADM-x3} is equivalent to
%Introducing a variable tau such that tau = \rho_Basel.
%Then the subproblem becomes
\be\label{equ:opt-tau}
\min\limits_{\tau\in\R, x \in \R^n}\ \tau+ \frac{\sigma}{2} \| x - v^{(j)}\|^2,\ \text{s.t.}\ \rho_{\BaselG}(x)\leq \tau,
\ee
whose solution can be obtained in a way similar to that in Lemmas \ref{lemma:opt-ADM-basel-x} and  \ref{lemma:opt-ADM-basel3-x} except that $C_0$ in
\eqref{eq:QP} and \eqref{eq:ADM-x-CVaR-solu} should be replaced by $\tau$ and $\tau$ should be added to the objective functions and $\tau$ should be included as an additional optimization variable in the minimization. The subproblem \eqref{eq:ADM-u3} can be solved by a standard QP solver.

% It can also be proved that solving \eqref{eq:ADM-x2} is equivalent to solving a QP. When $\rho(x(u))=\rho_{\VaR}(x(u))$, the closed-form solution of \eqref{eq:ADM-x2} is given by \eqref{eq:solution}, where $\gamma_{i^*}$ is set to $b_0$. %  simpler than that in Lemma \ref{lemma:opt-ADM-y-var}.

\item[\textbf{ADM for Solving Problem \eqref{equ:optc}:}]  %The counterpart of the ADM scheme \eqref{eq:ADM-x}-\eqref{eq:ADM-lmb} to problem  \eqref{equ:optc} is
%In fact, let $x=-\tR u$, we have
%\be\label{equ:optc_v2}
% \begin{aligned}
%\min_{u \in \U, \; x \in \R^n }\ \  & -\mu^\top u,\\
%\text{s.t. }\; & \rho(x) \leq C_0, \\
%& x + \tR u = 0.
%\end{aligned}\ee
The  augmented
Lagrangian function  for \eqref{equ:optc} is defined  as
\be
\Lc_c(x,u,\lambda):= -\mu^\top u+ \lambda^\top (x + \tR u) + \frac{\sigma}{2} \| x + \tR u \|^2,
\label{auglang-v2}
\ee
where $\sigma>0$ is the penalty parameter and $\lambda \in \R^n$ is the
Lagrangian multiplier. The  ADM method is
%In fact, we have:
\bea \label{eq:ADM-x2}
x^{(j+1)}  &=& \arg \min_{x \in \R^n}  \;
  \| x - v^{(j)}\|^2, \text{ s.t. } \rho(x) \le b_0, \\
\label{eq:ADM-u2}
 u^{(j+1)}  &=&  \arg \min_{u \in \U } \;   \frac{1}{2} u^\top \tR^\top \tR u +
 b_c^\top u,
  \\
 \label{eq:ADM-lmb2}
 \lambda^{(j+1)} &=& \lambda^{(j)} + \beta \sigma (x^{(j+1)}+ \tR u^{(j+1)}),
\eea
where $v^{(j)} = - \left(\tR u^{(j)} + \frac{1}{\sigma} \lambda^{(j)} \right)$,
$b_c=\tR^\top (\frac{1}{\sigma} \lambda^{(j)} + x^{(j+1)}) -
\frac{\mu}{\sigma}$,  and $\beta>0$.
For $\rho=\rho_{\Variance}$, \eqref{eq:ADM-x2} is a QP problem;
for $\rho=\rho_{\CVaR}$, \eqref{eq:ADM-x2} can be formulated as a QP problem by using \eqref{equ:cvar_opt_rep}; for $\rho=\rho_{\VaR}$, by an argument similar to the proof of Lemma \ref{lemma:opt-ADM-y-var}, the closed-form solution of
 \eqref{eq:ADM-x2} is given by \eqref{eq:solution} with $\gamma_{i^*}$ being replaced by $b_0$; for $\rho=\rho_{\BaselG}$, the solution to \eqref{eq:ADM-x2} can be obtained by Lemmas \ref{lemma:opt-ADM-basel-x} and  \ref{lemma:opt-ADM-basel3-x}.
 %  simpler than that in Lemma \ref{lemma:opt-ADM-y-var}.
 The subproblem \eqref{eq:ADM-u2}
is a standard QP problem.

Convergence results similar to  Theorem \ref{thm:cvg-feasi} can
 be established for  the ADM for models \eqref{equ:opt} and  \eqref{equ:optc}.
\end{description}

\section{Numerical Results}\label{sec:numerical}

In this section, we conduct computational experiments to demonstrate the effectiveness of
the ADM method for solving the mean-$\rho$-Basel model
using both simulated and real market data. In particular, we compare the performance of ADM method with that of
MIP/QP/LP solvers in CPLEX 12.4.
The numerical results suggest that the ADM method is promising in generating
solutions of high quality to the model in reasonable computational time.

\subsection{Data Description}
In our experiments, the real market data and simulated data sets are generated as follows.
\begin{itemize}
\item{\bf S\&P 500 Data Set.}
  The S\&P 500 data set comprises the daily returns of 359 stocks that have ever been included in the S\&P 500 index and do not have missing data during the following specified time periods. Let
$t_0=03/01/2012$. For $s=1, \ldots, 60$, $\tilde R_{SP}^{[s]}$ denotes the trailing
five-year daily returns of the stocks on day $t_0-s+1$ (i.e., the daily
returns of the stocks during the period from day %$t_0-s+1-252\times 5+1$
$t_0-s-2058$
to   day
$t_0-s+1$). Let $l=06/01/2007$ and $u=06/01/2009$. For $s=61, \ldots, 120$,
$\tilde R_{SP}^{[s]}$ is defined as the daily returns of the stocks during the
stressed period from  day $l+120-s$ to   day $u-s+61$. Then the S\&P data matrix $\tilde R_{SP}$ is defined from $\tilde R_{SP}^{[s]}, s=1, \ldots, 120$ by Eq. \eqref{equ:port_loss_obs}.
%If a stock has
%missing data during the above period, we exclude it from the portfolio; we end
%up having 359 stocks in the portfolio.
% Hence, for Basel 2.5, the number $n_s$ of historical observations is 504 and 463 for $s\in \{1,\ldots,60\}$ and $s\in \{61,\ldots,120\}$, respectively.

\item{\bf Simulated Data.} We simulate the  prices of 350 stocks based on a multi-dimensional version of the double-exponential jump diffusion model \citep{Kou-2002}:
\begin{equation}\label{equ:xh_4}
  \frac{dS_i(t)}{S_i(t-)}=\mu_i dt + \sigma_i dW_i(t) + d\left(\sum_{k=1}^{N_i(t)}(e^{V_{ik}}-1)\right), i=1, \ldots, n,
\end{equation}
where $W_1(t), \ldots, W_n(t)$ are $n$ correlated Brownian motions with $dW_i(t)dW_j(t)=\rho_{ij}dt$; $N_i(t)$ is a Poisson process with intensity $\lambda_i$; $N_i(t)$ is independent of $N_j(t)$ for $i\neq j$; $\{V_{i1}, V_{i2}, \ldots\}$ are i.i.d. log jump sizes with a double-exponential probability density function $f_i(x)=
    p_i\eta_{iu} e^{-\eta_{iu} x}1_{\{x\geq 0\}}+(1-p_i)\eta_{id} e^{\eta_{id} x}1_{\{x< 0\}}$;
    %where $0<p_i<1, \eta_{iu}>0, \eta_{id}>0;
$V_{ik}$ and $V_{jl}$ are independent for $i\neq j$; and the Brownian motions, Poisson processes, and  jump sizes are mutually independent.
%$(W_1(t), W_2(t), \ldots, W_n(t))$ are independent of $(N_1(t), N_2(t),\ldots, N_n(t))$,  and the jump sizes $\{V_{i1}, V_{i2}, \ldots\}$.
%Under  model \eqref{equ:xh_4}, the log return of the $i$th stock from time $t$ to $t+\Delta t$ is given by
%\[  (\mu_i - \frac{1}{2}\sigma_i^2) \Delta t + \sigma_i (W_i(t+ \Delta t) - W_i(t)) + \sum_{k=N_i(t)+1}^{N_i(t+\Delta t)}V_{ik}.
%\]
%which has the same distribution as
%\[  (\mu - \frac{1}{2}\sigma^2) \Delta t + \sigma W(\Delta t) + \sum_{i=1}^{N(\Delta t)}V_i.
%\]
The stock returns generated in the above model have the same tail heaviness as those generated by the negative exponential tail model considered in \citet*{Lim2011163}. Two sets of parameters $\{\mu_i,
\sigma_i, \lambda_i, p_i, \eta_{iu}, \eta_{id}, \rho_{ij}\}$ are used to simulate stock returns under normal and stressed market conditions, respectively; these parameters are estimated from the historical data of some large-cap stocks during normal and stressed market conditions, respectively. $\Delta t$ is set to be $1/252$ (one day).

%The values of these parameters as well as the Matlab codes for simulating the data are available upon request.
% at
%  \url{http://www.math.ust.hk/~maxhpeng}.
%\[  (\mu, \sigma, \lambda, p, \eta_u, \eta_d) = (0.0007, 0.0047, 1.0264, 0.4521, 174.09, 185.92).\]

\end{itemize}

%The above procedures produce the return observations $\tR$.

\subsection{Parameter Settings of the ADM and MIP}\label{subsec:parameters_adm}
 Our method is implemented in MATLAB. All the
experiments were performed on a Dell Precision Workstation T5500 with Intel Xeon CPU E5620 at 2.40GHz and 12GB of memory
running Ubuntu 12.04 and MATLAB 2011b.  All the quadratic programming subproblems in
the ADM method are solved using the QP solvers in CPLEX 12.4 with Matlab interface; and the mixed integer programming (MIP) reformulations of the asset allocation models are solved using the MIP solvers in CPLEX 12.4. In
our test, the parameters $\sigma$ and $\beta$ in \eqref{eq:ADM-x}-\eqref{eq:ADM-lmb}
are set to be $10^{-3}$ and $0.1$, respectively. The initial Lagrangian
multipliers are  $\lambda^{(0)}  = 0$ and $\pi^{(0)}=0$. The method is terminated if either $\|x^{(j+1)}+\tilde{R}u^{(j+1)}\|^2+\|y^{(j+1)}+\tilde{Y}u^{(j+1)}\|^2 \le 10^{-8}$,
$ \frac{\|u^{(j+1)}- u^{(j)}\|}{\max(1, \| u^{(j)}\|)} \le 10^{-4}$, or
the number of
iterations has reached an upper bound of 2000. The default setting of the MIP solver in CPLEX
12.4 is used.  The maximum CPU
time limit for all solvers is set  to 3600 seconds.

\subsection{Comparing ADM with MIP/QP on the Mean-Variance-Basel Model}
\label{subsec:compare_mean_variance}

In this subsection, we evaluate the performance of the ADM on the
mean-variance-Basel problems:
%respectively, as
\be\label{equ:Variance-Basel-prob}
 \begin{aligned}
\min_{u \in \Ur }\;  & \rho_{\Variance}(y(u))\\
\text{s.t. }\;\; & \rho_{\BaselTwoDotFive}(x(u)) \leq C_0,\end{aligned} \quad \mbox{ and }
\quad \begin{aligned}
\min_{u \in \Ur }\;  & \rho_{\Variance}(y(u))\\
\text{s.t. }\;\; & \rho_{\BaselThree}(x(u)) \leq C_0.\end{aligned} \ee
%where
%$y(u)=-\tilde Yu$ and $x(u)=-\tR u$.
The mean-variance-Basel2.5 problem can be solved using the MIP method, and the
mean-variance-Basel3 problem is a QP problem that can be solved using the QP solver in CPLEX 12.4.

%We compare the ADM with the MIP for the mean-variance-Basel-2.5 problem; and compare the ADM with the QP for the mean-variance-Basel-III problem.

We compare the ADM with the MIP/QP methods for the two problems, respectively, for different numbers of stocks $d\in \{100, 150, 200, 250,
300, 350 \}$ using real market and simulated data. For the real market data, $\tR$ is defined to be a submatrix of $\tilde R_{SP}$ consisting of
$d$ columns of $\tilde R_{SP}$ that are randomly selected.
The mean $\mu$ used for defining $\Ur$ is set as the sample mean of $\tR$. The prescribed
return level $r_0$ is set to be the 80\% quantile of the cross-sectional expected returns of the $d$ stocks. $\tilde Y$ is obtained by deleting the
 duplicated rows in $\tR$. The parameters in \eqref{equ:basel_3_sta} are set at
 $\alpha=0.99$, $k=3$, and $\ell=3$; those in \eqref{equ:basel_3_cvar} are set at $\alpha=0.98$ and $\ell=6$. $C_0$ is set at $0.2$.
% The confidence level of $\rho_{\VaR}$ in the formula of $\rho_{\BaselTwoDotFive}$ is set to be $0.99$; and $C_0$ is set to be $0.2$. The confidence level for $\rho_{\BaselThree}$ is set to be 0.98.
%We consider the following two cases: i) solving model \eqref{equ:VaR-Basel-prob} with $\rho_{\BaselTwoDotFive}(x(u))$ in the constraints by using the ADM  and MIP method, respectively; ii) solving model \eqref{equ:VaR-Basel-prob} with   $\rho_{\BaselThree}(x(u))$ in the constraints by using the ADM  and MIP method, respectively. \alert{The confidence level for $\rho_{\BaselThree}$ is set to be 0.98.}
Table \ref{tab:variable-Variance-Basel} reports the numbers of binary variables, continuous
variables, and linear constraints, denoted by ``binary,'' ``continuous,'' and
``constraints,'' respectively, of the two problems in \eqref{equ:Variance-Basel-prob}.

\begin{table}\caption{The number of binary variables, continuous
variables, and linear constraints in the
MIP/QP formulation of the mean-variance-Basel problems.}
%mean-variance-Basel2.5 and the mean-variance-Basel3 problems.}
\label{tab:variable-Variance-Basel}
\centering
\begin{tabular}{cccccccccccccccccccc}
\hline
&& \multicolumn{3}{c}{$\rho_{\BaselTwoDotFive}(x(u)) \le C_0$}&&\multicolumn{3}{c}{$\rho_{\BaselThree}(x(u))\le C_0$}\\
\cline{3-5}\cline{7-9}
$d$ && binary & continuous & constraints &&binary & continuous & constraints  \\
\hline
100&& 58020 &4601 & 62526 && -- &32319 & 32163 \\
150&& 58020 &4651 & 62526 && -- &32369 & 32163 \\
200&& 58020 &4701 & 62526 && -- &32419 & 32163 \\
250&& 58020 &4751 & 62526 && -- &32469 & 32163 \\
300&& 58020 &4801 & 62526 && -- &32519 & 32163 \\
350&& 58020 &4851 & 62526 && -- &32569 & 32163 \\
\hline
\end{tabular}
\end{table}

The optimal objective value $\rho_\Variance(y(u))$ obtained and the CPU time used by the ADM and MIP/QP methods for the simulated and real market data are presented in Figures \ref{fig:Variance-Basel-simu} and \ref{fig:Variance-Basel-sp}, respectively. These values, as well as $\rho_{\BaselTwoDotFive}(x(u))$ and $\rho_{\BaselThree}(x(u))$, are
reported in Tables \ref{tab:Variance-Basel-simu}
 and \ref{tab:Variance-Basel-sp}.  We can observe that the ADM
obtains a better objective value of $\rho_{\Variance}$ and is faster than the MIP/QP methods.

\begin{figure}[ht]
\centering
 \includegraphics[width=0.8\textwidth]{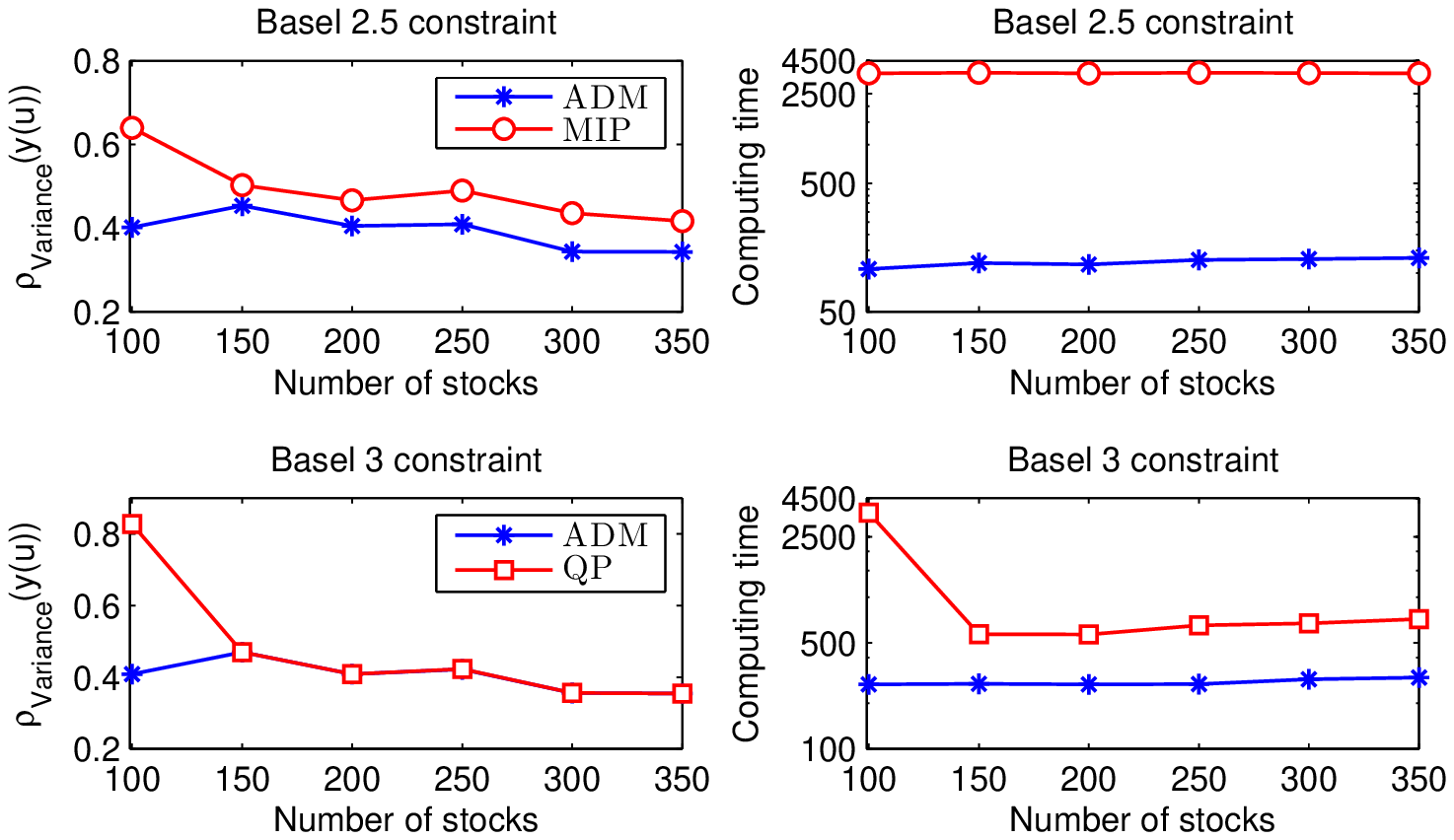}
 \caption{\label{fig:Variance-Basel-simu}Comparing the ADM with the MIP for the mean-variance-Basel2.5 problem and comparing the ADM with the QP for the mean-variance-Basel3 problem for different numbers of stocks
using simulated data. CPU time is expressed in seconds.}
\end{figure}

\begin{figure}[ht]
\centering
 \includegraphics[width=0.8\textwidth]{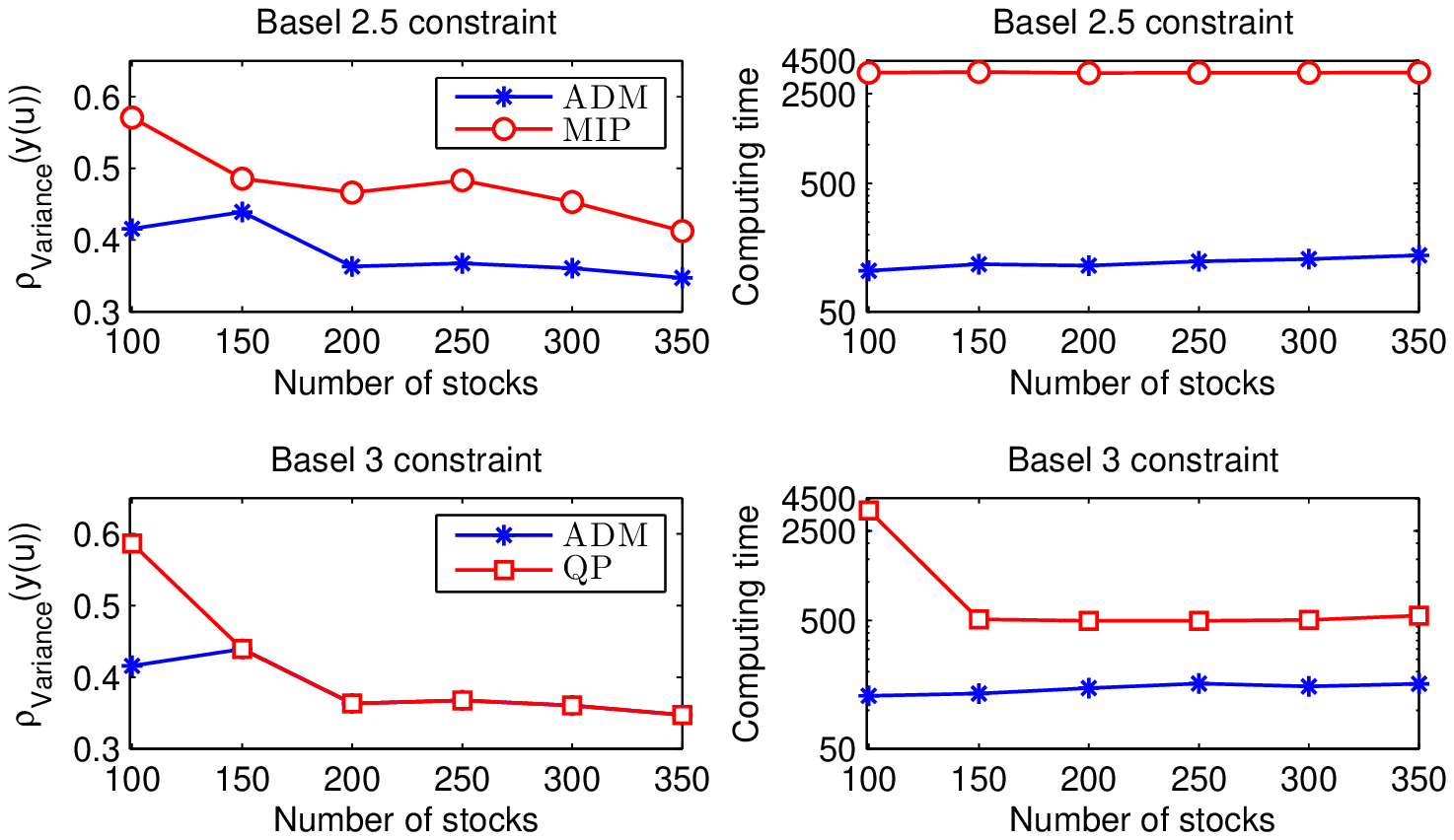}
\caption{\label{fig:Variance-Basel-sp}Comparing the ADM with the MIP for the mean-variance-Basel2.5 problem and comparing the ADM with the QP for the mean-variance-Basel3 problem for different numbers of stocks
using real market data. CPU time is expressed in seconds.}
\end{figure}

\begin{table}\caption{The numerical results of solving the mean-variance-Basel problems with
simulated data using the ADM and the MIP/QP methods.}
  \label{tab:Variance-Basel-simu}
 \setlength{\tabcolsep}{1.2pt}
%\resizebox{\linewidth}{!}{
\centering
\begin{tabular}{ccccccccccccccccccccccccccccc}
\hline
\multicolumn{1}{c}{stocks} && \multicolumn{3}{c}{ $\rm ADM_{\BaselTwoDotFive\le C_0}$}&&
\multicolumn{3}{c}{ $\rm MIP_{\BaselTwoDotFive\le C_0}$} && \multicolumn{3}{c}{ $\rm
ADM_{Basel3\le C_0}$}&&  \multicolumn{3}{c}{ $\rm QP_{Basel3\le C_0}$} \\
\cline{3-5}\cline{7-9}\cline{11-13}\cline{15-17}
$d$  && $\rho_{\Variance}$ & time & $\rho_{\rm \BaselTwoDotFive}$ && $\rho_{\Variance}$ &
time &  $\rho_{\rm \BaselTwoDotFive}$ && $\rho_{\Variance}$ & time & $\rho_{\rm Basel3}$
&& $\rho_{\Variance}$ & time &  $\rho_{\rm Basel3}$\\
\hline
100 && 0.4020 & 108 & 0.146 && 0.6399 & 3600 & 0.158 && 0.4088 & 266 & 0.200 && 0.8277 & 3603 & 0.144 &\\
\hline
150 && 0.4538 & 120 & 0.164 && 0.5029 & 3631 & 0.171 && 0.4701 & 268 & 0.200 && 0.4702 & 572 & 0.200 &\\
\hline
200 && 0.4054 & 117 & 0.164 && 0.4664 & 3605 & 0.147 && 0.4087 & 267 & 0.200 && 0.4087 & 567 & 0.200 &\\
\hline
250 && 0.4090 & 127 & 0.161 && 0.4895 & 3630 & 0.153 && 0.4226 & 268 & 0.200 && 0.4227 & 653 & 0.200 &\\
\hline
300 && 0.3437 & 128 & 0.151 && 0.4356 & 3613 & 0.154 && 0.3561 & 288 & 0.200 && 0.3562 & 671 & 0.200 &\\
\hline
350 && 0.3428 & 132 & 0.156 && 0.4169 & 3605 & 0.163 && 0.3543 & 296 & 0.200 && 0.3544 & 717 & 0.200 &\\
\hline
\end{tabular}
  \end{table}

\begin{table}\caption{The numerical results of solving the mean-variance-Basel problems with
real market data using the ADM and the MIP/QP methods.}
  \label{tab:Variance-Basel-sp}
 \setlength{\tabcolsep}{1.2pt}
%\resizebox{\linewidth}{!}{
\centering
\begin{tabular}{ccccccccccccccccccccccccccccc}
\hline
\multicolumn{1}{c}{stocks} && \multicolumn{3}{c}{ $\rm ADM_{\BaselTwoDotFive\le C_0}$}&&
\multicolumn{3}{c}{ $\rm MIP_{\BaselTwoDotFive\le C_0}$} && \multicolumn{3}{c}{ $\rm
ADM_{Basel3\le C_0}$}&&  \multicolumn{3}{c}{ $\rm QP_{Basel3\le C_0}$} \\
\cline{3-5}\cline{7-9}\cline{11-13}\cline{15-17}
$d$  && $\rho_{\Variance}$ & time & $\rho_{\rm \BaselTwoDotFive}$ && $\rho_{\Variance}$ &
time &  $\rho_{\rm \BaselTwoDotFive}$ && $\rho_{\Variance}$ & time & $\rho_{\rm Basel3}$
&& $\rho_{\Variance}$ & time &  $\rho_{\rm Basel3}$\\
\hline
100 && 0.4157 & 104 & 0.145 && 0.5710 & 3627 & 0.148 && 0.4157 & 130 & 0.157 && 0.5867 & 3610 & 0.122 &\\
\hline
150 && 0.4393 & 118 & 0.155 && 0.4859 & 3698 & 0.136 && 0.4393 & 135 & 0.145 && 0.4393 & 510 & 0.145 &\\
\hline
200 && 0.3633 & 114 & 0.133 && 0.4662 & 3619 & 0.135 && 0.3633 & 149 & 0.141 && 0.3633 & 497 & 0.141 &\\
\hline
250 && 0.3678 & 124 & 0.145 && 0.4832 & 3630 & 0.139 && 0.3678 & 162 & 0.152 && 0.3678 & 497 & 0.152 &\\
\hline
300 && 0.3608 & 129 & 0.138 && 0.4530 & 3638 & 0.148 && 0.3608 & 153 & 0.144 && 0.3608 & 505 & 0.144 &\\
\hline
350 && 0.3473 & 137 & 0.137 && 0.4125 & 3649 & 0.138 && 0.3473 & 161 & 0.147 && 0.3473 & 547 & 0.147 &\\
\hline
\end{tabular}
  \end{table}

\subsection{Comparing ADM with MIP/LP on the Mean-CVaR-Basel Model}

In this subsection, we evaluate the performance of the ADM on the
mean-CVaR-Basel problems:
\be\label{equ:CVaR-Basel-prob}
 \begin{aligned}
\min_{u \in \Ur }\;  & \rho_{\CVaR}(y(u))\\
\text{s.t. }\;\; & \rho_{\BaselTwoDotFive}(x(u)) \leq C_0,\end{aligned} \quad \mbox{ and }
\quad \begin{aligned}
\min_{u \in \Ur }\;  & \rho_{\CVaR}(y(u))\\
\text{s.t. }\;\; & \rho_{\BaselThree}(x(u)) \leq C_0.\end{aligned} \ee
The mean-CVaR-Basel2.5 problem can be solved using the MIP method, and the
mean-CVaR-Basel3 problem can be formulated as a LP problem and solved using the dual simplex (LP) solver in CPLEX 12.4.

We compare the ADM with the MIP/LP methods for the two problems, respectively. The setup of the experiments is the same as in Section \ref{subsec:compare_mean_variance}. Hence, the numbers of binary variables and linear constraints are the same as those in Table \ref{tab:variable-Variance-Basel}, and the number of continuous
variables is equal to that in Table
\ref{tab:variable-Variance-Basel} plus one.

%\begin{table}\caption{Problem characteristics of
% mean-CVaR-Basel   model  with $\rho_{\BaselTwoDotFive}$ or
%  $\rho_{\BaselThree}$ constraints.}\label{tab:variable-CVaR-Basel}
%\centering
%\begin{tabular}{cccccccccccccccccccc}
%\hline
%&& \multicolumn{3}{c}{$\rho_{\rm Basel2.5}(x(u)) \le C_0$}&&\multicolumn{3}{c}{$\rho_{\rm Basel3}(x(u))\le C_0$}\\
%\cline{3-5}\cline{7-9}
%$d$ && binary & continuous & constraints &&binary & continuous & constraints  \\
%\hline
%100&& 58020 &4602 & 62526 && -- &32320 & 32163 \\
%150&& 58020 &4652 & 62526 && -- &32370 & 32163 \\
%200&& 58020 &4702 & 62526 && -- &32420 & 32163 \\
%250&& 58020 &4752 & 62526 && -- &32470 & 32163 \\
%300&& 58020 &4802 & 62526 && -- &32520 & 32163 \\
%350&& 58020 &4852 & 62526 && -- &32570 & 32163 \\
%\hline
%\end{tabular}
%\end{table}

The optimal objective value $\rho_\CVaR(y(u))$ obtained and the CPU time used by the ADM and MIP/LP methods for the simulated and real market data are presented in Figures \ref{fig:CVaR-Basel-simu} and \ref{fig:CVaR-Basel-sp}, respectively. These values, as well as $\rho_{\BaselTwoDotFive}(x(u))$ and $\rho_{\BaselThree}(x(u))$, are
reported in Tables \ref{tab:CVaR-Basel-simu}
 and \ref{tab:CVaR-Basel-sp}. We can observe that (i) the ADM
method is faster than the MIP method for the mean-CVaR-Basel2.5 problem but is slower than the LP method for the mean-CVaR-Basel3 problem; (ii)
the optimal objective value $\rho_\CVaR(y(u))$ obtained by the ADM and the MIP/LP are almost the same. In fact, the largest absolute value of the relative difference (``rel.dif'') of
$\rho_{\CVaR}$ between that obtained using the ADM and that obtained using the MIP/LP is $0.39\%$, where ``rel.dif'' is  defined by
$ \mbox{rel.dif}:= (\rho_{\CVaR}(y(u_{\rm ADM})) -
\rho_{\CVaR}(y(u_{\rm MIP/LP})))/\rho_{\CVaR}(y(u_{\rm MIP/LP}))$.

\begin{figure}[ht]
\centering
 \includegraphics[width=0.8\textwidth]{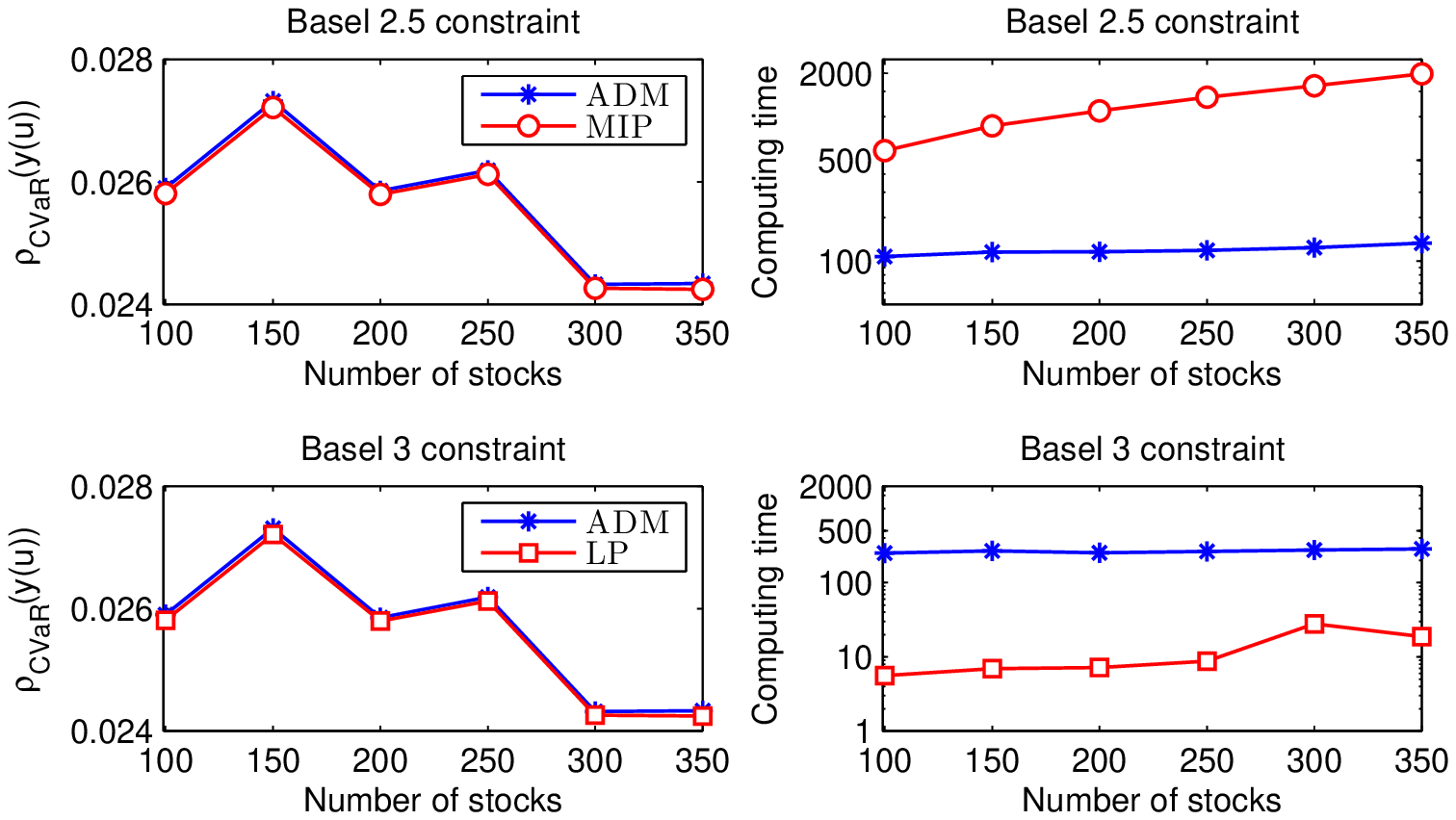}
 \caption{\label{fig:CVaR-Basel-simu}
Comparing the ADM with the MIP for the mean-CVaR-Basel2.5 problem and comparing the ADM with the LP for the mean-CVaR-Basel3 problem for different numbers of stocks
using simulated data. CPU time is expressed in seconds.}
\end{figure}

\begin{figure}[ht]
\centering
 \includegraphics[width=0.8\textwidth]{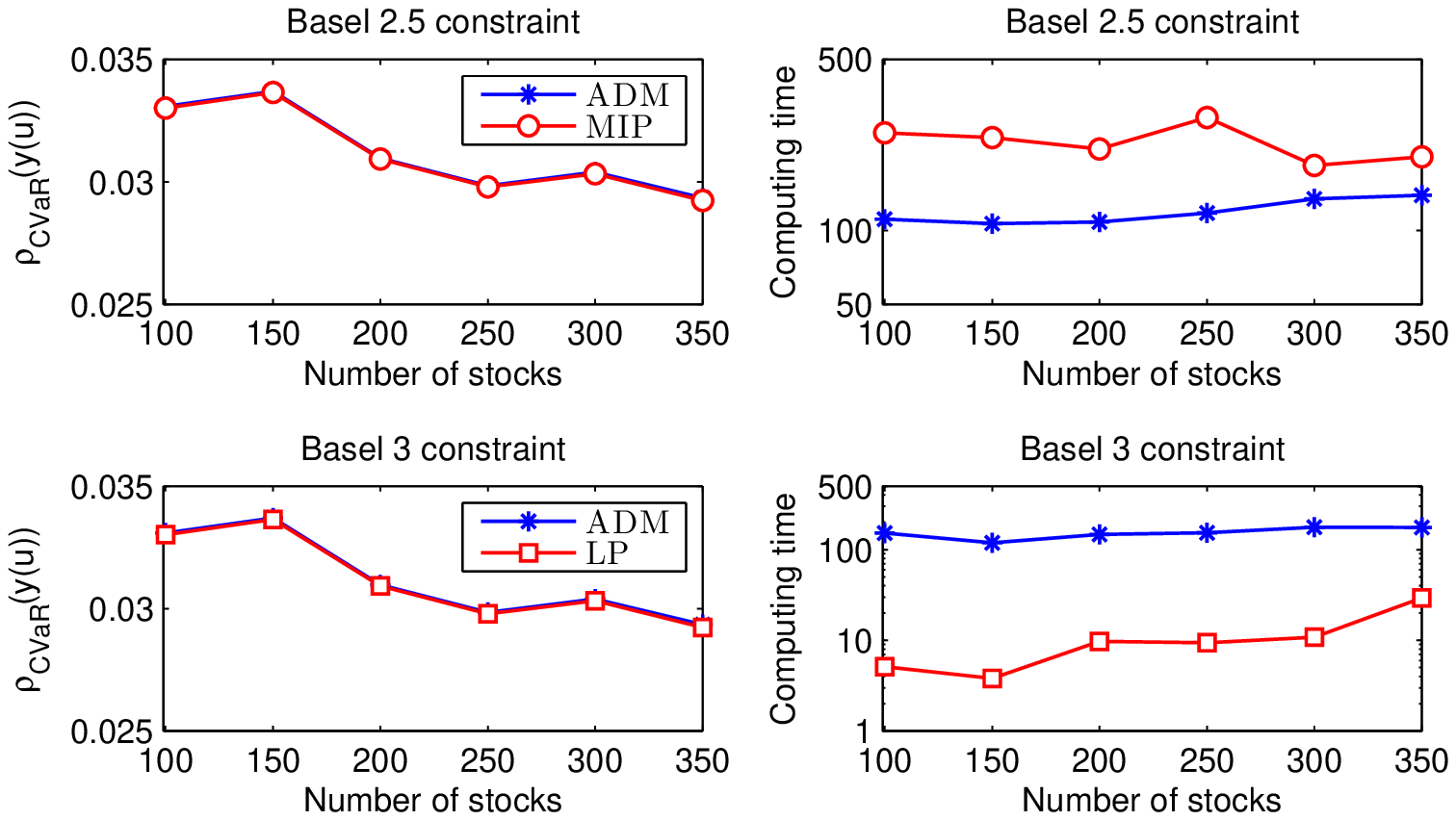}
\caption{\label{fig:CVaR-Basel-sp}Comparing the ADM with the MIP for the mean-CVaR-Basel2.5 problem and comparing the ADM with the LP for the mean-CVaR-Basel3 problem for different numbers of stocks
using real market data. CPU time is expressed in seconds.}
\end{figure}

\begin{table}\caption{
The numerical results of solving the mean-CVaR-Basel problems with
simulated data using the ADM and the MIP/LP methods.}
  \label{tab:CVaR-Basel-simu}
 \setlength{\tabcolsep}{1.5pt}
%\resizebox{\linewidth}{!}{
\centering
\begin{tabular}{ccccccccccccccccccccccccccccc}
\hline
\multicolumn{1}{c}{stocks} && \multicolumn{3}{c}{ $\rm ADM_{\BaselTwoDotFive\le C_0}$}&&
\multicolumn{3}{c}{ $\rm MIP_{\BaselTwoDotFive\le C_0}$} && \multicolumn{3}{c}{ $\rm
ADM_{Basel3\le C_0}$}&&  \multicolumn{3}{c}{ $\rm LP_{Basel3\le C_0}$} \\
\cline{3-5}\cline{7-9}\cline{11-13}\cline{15-17}
$d$  && $\rho_{\CVaR}$ & time & $\rho_{\rm \BaselTwoDotFive}$ && $\rho_{\CVaR}$ & time &
$\rho_{\rm \BaselTwoDotFive}$ && $\rho_{\CVaR}$ & time & $\rho_{\rm Basel3}$ && $\rho_{\CVaR}$ & time &  $\rho_{\rm Basel3}$\\
\hline
100 && 0.0259 & 108 & 0.135 && 0.0258 & 583 & 0.134 && 0.0259 & 252 & 0.194 && 0.0258 & 6 & 0.197 &\\
\hline
150 && 0.0273 & 115 & 0.151 && 0.0272 & 866 & 0.151 && 0.0273 & 268 & 0.198 && 0.0272 & 7 & 0.194 &\\
\hline
200 && 0.0259 & 116 & 0.141 && 0.0258 & 1099 & 0.143 && 0.0259 & 252 & 0.170 && 0.0258 & 7 & 0.172 &\\
\hline
250 && 0.0262 & 118 & 0.146 && 0.0261 & 1365 & 0.142 && 0.0262 & 264 & 0.184 && 0.0261 & 9 & 0.185 &\\
\hline
300 && 0.0243 & 124 & 0.132 && 0.0243 & 1640 & 0.132 && 0.0243 & 276 & 0.180 && 0.0243 & 28 & 0.184 &\\
\hline
350 && 0.0243 & 133 & 0.134 && 0.0242 & 1981 & 0.131 && 0.0243 & 284 & 0.178 && 0.0242 & 19 & 0.181 &\\
\hline
\end{tabular}
  \end{table}

\begin{table}\caption{The numerical results of solving the mean-CVaR-Basel problems with
real market data using the ADM and the MIP/LP methods.}
  \label{tab:CVaR-Basel-sp}
 \setlength{\tabcolsep}{1.5pt}
%\resizebox{\linewidth}{!}{
\centering
\begin{tabular}{ccccccccccccccccccccccccccccc}
\hline
\multicolumn{1}{c}{stocks} && \multicolumn{3}{c}{ $\rm ADM_{\BaselTwoDotFive\le C_0}$}&&
\multicolumn{3}{c}{ $\rm MIP_{\BaselTwoDotFive\le C_0}$} && \multicolumn{3}{c}{ $\rm
ADM_{Basel3\le C_0}$}&&  \multicolumn{3}{c}{ $\rm LP_{Basel3\le C_0}$} \\
\cline{3-5}\cline{7-9}\cline{11-13}\cline{15-17}
$d$  && $\rho_{\CVaR}$ & time & $\rho_{\rm \BaselTwoDotFive}$ && $\rho_{\CVaR}$ & time &  $\rho_{\rm \BaselTwoDotFive}$ && $\rho_{\CVaR}$ & time & $\rho_{\rm Basel3}$ && $\rho_{\CVaR}$ & time &  $\rho_{\rm Basel3}$\\
\hline
100 && 0.0331 & 111 & 0.141 && 0.0330 & 251 & 0.143 && 0.0331 & 153 & 0.156 && 0.0330 & 5 & 0.160 &\\
\hline
150 && 0.0337 & 107 & 0.133 && 0.0336 & 240 & 0.134 && 0.0337 & 119 & 0.143 && 0.0336 & 4 & 0.139 &\\
\hline
200 && 0.0310 & 109 & 0.132 && 0.0309 & 215 & 0.139 && 0.0310 & 147 & 0.141 && 0.0309 & 10 & 0.147 &\\
\hline
250 && 0.0298 & 118 & 0.142 && 0.0298 & 289 & 0.145 && 0.0298 & 153 & 0.151 && 0.0298 & 9 & 0.154 &\\
\hline
300 && 0.0304 & 135 & 0.140 && 0.0303 & 184 & 0.141 && 0.0304 & 177 & 0.147 && 0.0303 & 11 & 0.153 &\\
\hline
350 && 0.0293 & 140 & 0.140 && 0.0292 & 200 & 0.142 && 0.0293 & 175 & 0.151 && 0.0292 & 29 & 0.150 &\\
\hline
\end{tabular}
  \end{table}

\subsection{Comparing ADM with MIP on the Mean-VaR-Basel Model}
\label{sec:num-model}

In this subsection, we compare the performance of the ADM with that of the MIP on the
mean-VaR-Basel models:
\be\label{equ:VaR-Basel-prob}
 \begin{aligned}
\min_{u \in \Ur }\;  & \rho_{\VaR}(y(u))\\
\text{s.t. }\;\; & \rho_{\BaselTwoDotFive}(x(u)) \leq C_0,\end{aligned} \quad \mbox{ and }
\quad \begin{aligned}
\min_{u \in \Ur }\;  & \rho_{\VaR}(y(u))\\
\text{s.t. }\;\; & \rho_{\BaselThree}(x(u)) \leq C_0.\end{aligned} \ee
The setup of the experiments is the same as that in Section \ref{subsec:compare_mean_variance}.
Table \ref{tab:variable-VaR-Basel} reports the number of binary variables, continuous
variables, and linear constraints in the MIP formulation of the problems in \eqref{equ:VaR-Basel-prob}.

%We compare the two methods for different number of stocks $d\in \{100, 150, 200, 250,
%300, 350 \}$ using both real market data and simulated data. For the real market data, $\tR$ is defined to be a submatrix of $\tilde R_{SP}$ consisting of
%$d$ columns of $\tilde R_{SP}$ that are randomly selected.
%The mean $\mu$ in $\Ur$ is set as the sample mean of $\tR$. The prescribed
%return level $r_0$ is set to be the 80\% quantile of the cross-sectional expected returns of the $d$ stocks. $\tilde Y$ is obtained by deleting the
% duplicated rows in $\tR$. The confidence level of $\rho_{\VaR}$ is set to be $0.99$ and $C_0$ is set to be $0.2$. The confidence level for $\rho_{\BaselThree}$ is set to be 0.98.
%%We consider the following two cases: i) solving model \eqref{equ:VaR-Basel-prob} with $\rho_{\BaselTwoDotFive}(x(u))$ in the constraints by using the ADM  and MIP method, respectively; ii) solving model \eqref{equ:VaR-Basel-prob} with   $\rho_{\BaselThree}(x(u))$ in the constraints by using the ADM  and MIP method, respectively. \alert{The confidence level for $\rho_{\BaselThree}$ is set to be 0.98.}
%Table \ref{tab:variable-VaR-Basel} report the number of binary variables, continuous
%variables, and linear constraints, denoted by ``binary'', ``continuous'', and
%``constraints'', respectively, in the MIP reformulation of \eqref{equ:VaR-Basel-prob}.

\begin{table}\caption{
The number of binary variables, continuous
variables, and linear constraints in the MIP formulation of the mean-VaR-Basel problems.}\label{tab:variable-VaR-Basel}
\centering
\begin{tabular}{cccccccccccccccccccc}
\hline
&& \multicolumn{3}{c}{$\rho_{\BaselTwoDotFive}(x(u)) \le C_0$}&&\multicolumn{3}{c}{$\rho_{\BaselThree}(x(u))\le C_0$}\\
\cline{3-5}\cline{7-9}
$d$ && binary & continuous & constraints &&binary & continuous & constraints  \\
\hline
100&& 62399 &223 & 62527 && 4379 &27941 & 32164 \\
150&& 62399 &273 & 62527 && 4379 &27991 & 32164 \\
200&& 62399 &323 & 62527 && 4379 &28041 & 32164 \\
250&& 62399 &373 & 62527 && 4379 &28091 & 32164 \\
300&& 62399 &423 & 62527 && 4379 &28141 & 32164 \\
350&& 62399 &473 & 62527 && 4379 &28191 & 32164 \\
\hline
\end{tabular}
\end{table}

The optimal objective value $\rho_\VaR(y(u))$ obtained and the CPU time used by the ADM and the MIP methods for the simulated and real market data are presented in Figures \ref{fig:VaR-Basel-simu} and  \ref{fig:VaR-Basel-sp}, respectively. These values, as well as $\rho_{\BaselTwoDotFive}(x(u))$ and $\rho_{\BaselThree}(x(u))$, are
reported in Tables \ref{tab:VaR-Basel-simu}
 and \ref{tab:VaR-Basel-sp}. The figures and tables show that the ADM is a very good alternative to
the MIP for the mean-VaR-Basel problems because: (i) The ADM is much faster than the MIP. (ii) For the mean-VaR-Basel2.5 problem, the optimal objective value $\rho_\VaR$
computed by the ADM is smaller than that computed by the MIP
except in two cases; in fact, the relative difference of $\rho_{\rm VaR}$ between the ADM and the MIP, which is defined by
$(\rho_{\rm VaR}(y(u_{\rm ADM})) - \rho_{\rm
VaR}(y(u_{\rm MIP})))/\rho_{\rm VaR}(y(u_{\rm MIP}))$, is in the range of
$[-19.65\text{\%}, 8.07\text{\%}]$, which shows that the ADM may be slightly inferior to the MIP in some cases but can be significantly preferable to the MIP in other cases. (iii) For the mean-VaR-Basel3 problem, the relative difference of $\rho_{\rm VaR}$ between the ADM and the MIP is in the range of $[-23.18\text{\%}, 5.44\text{\%}]$, which shows that overall the ADM achieves better objective value than the MIP.
\begin{figure}[ht]
\centering
 \includegraphics[width=0.8\textwidth]{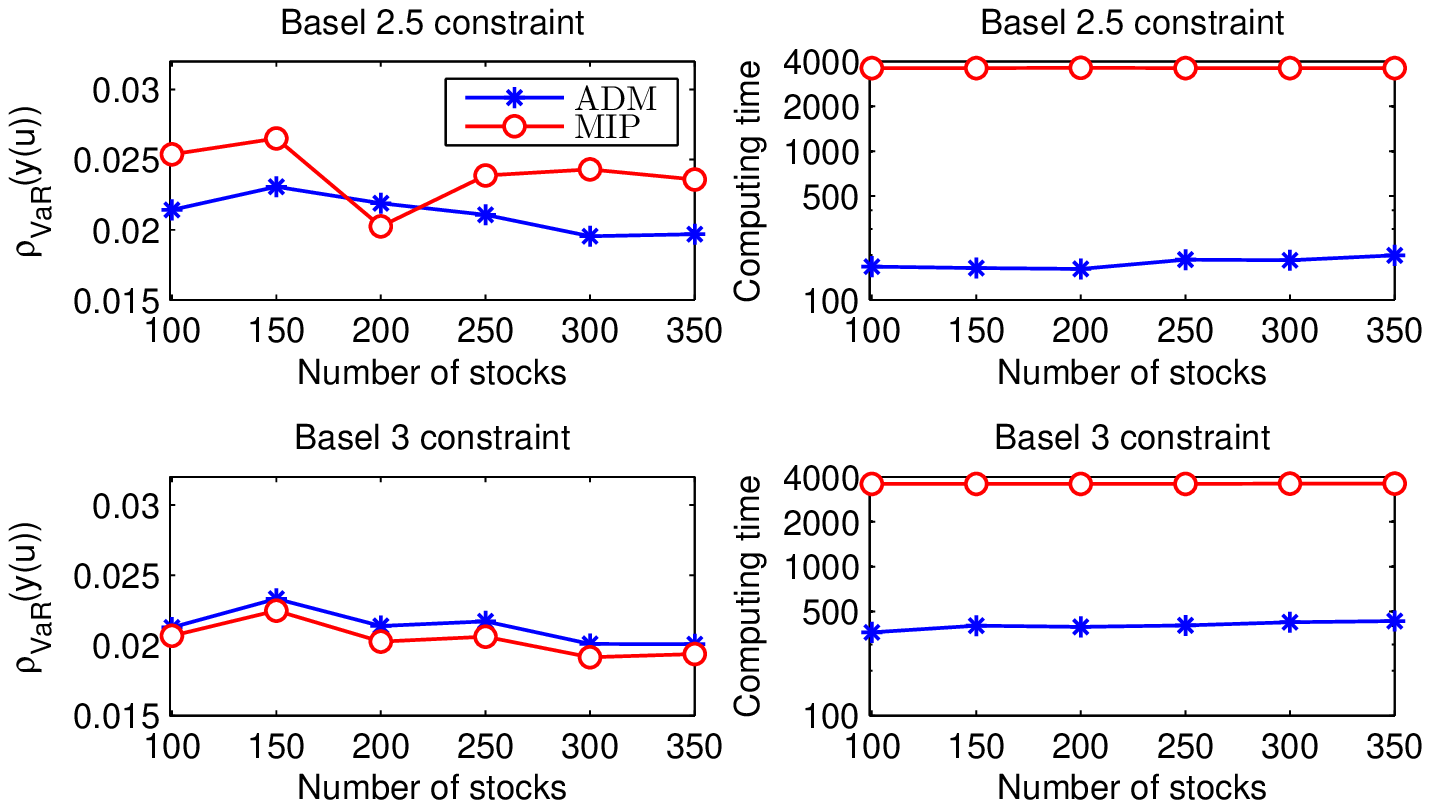}
 \caption{\label{fig:VaR-Basel-simu}
Comparing the ADM with the MIP for the mean-VaR-Basel problems for different numbers of stocks using simulated data. CPU time is expressed in seconds.}
\end{figure}

\begin{figure}[ht]
\centering
 \includegraphics[width=0.8\textwidth]{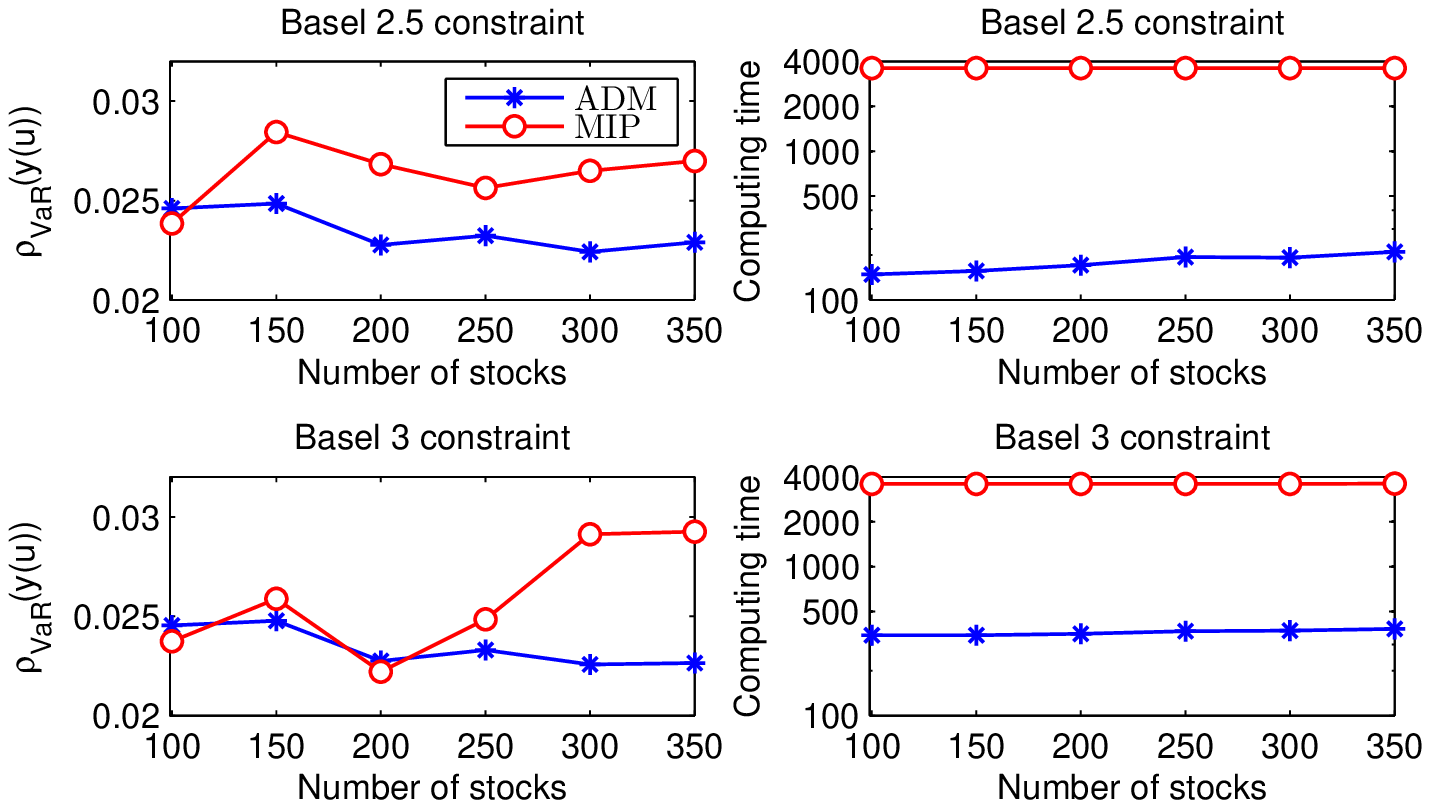}
\caption{\label{fig:VaR-Basel-sp}Comparing the ADM with the MIP for the mean-VaR-Basel problems for different numbers of stocks using real market data. CPU time is expressed in seconds.}
\end{figure}

\begin{table}\caption{The numerical results obtained when solving the mean-VaR-Basel problems with
simulated data using the ADM and the MIP methods.}
  \label{tab:VaR-Basel-simu}
 \setlength{\tabcolsep}{1.5pt}
%\resizebox{\linewidth}{!}{
\centering
\begin{tabular}{ccccccccccccccccccccccccccccc}
\hline
\multicolumn{1}{c}{stocks} && \multicolumn{3}{c}{ $\rm ADM_{\BaselTwoDotFive\le C_0}$}&&  \multicolumn{3}{c}{ $\rm MIP_{\BaselTwoDotFive\le C_0}$} && \multicolumn{3}{c}{ $\rm ADM_{Basel3\le C_0}$}&&  \multicolumn{3}{c}{ $\rm MIP_{Basel3\le C_0}$} \\
\cline{3-5}\cline{7-9}\cline{11-13}\cline{15-17}
$d$  && $\rho_{\rm VaR}$ & time & $\rho_{\rm \BaselTwoDotFive}$ && $\rho_{\rm VaR}$ & time &  $\rho_{\rm \BaselTwoDotFive}$ && $\rho_{\rm VaR}$ & time & $\rho_{\rm Basel3}$ && $\rho_{\rm VaR}$ & time &  $\rho_{\rm Basel3}$\\
\hline
100 && 0.0214 & 174 & 0.151 && 0.0254 & 3602 & 0.158 && 0.0213 & 361 & 0.200 && 0.0207 & 3601 & 0.200 &\\
\hline
150 && 0.0231 & 162 & 0.174 && 0.0265 & 3602 & 0.171 && 0.0233 & 397 & 0.200 && 0.0225 & 3601 & 0.200 &\\
\hline
200 && 0.0219 & 160 & 0.170 && 0.0202 & 3607 & 0.161 && 0.0214 & 390 & 0.200 && 0.0203 & 3601 & 0.200 &\\
\hline
250 && 0.0210 & 175 & 0.166 && 0.0239 & 3604 & 0.151 && 0.0217 & 402 & 0.200 && 0.0206 & 3602 & 0.200 &\\
\hline
300 && 0.0195 & 180 & 0.162 && 0.0243 & 3611 & 0.153 && 0.0201 & 422 & 0.200 && 0.0192 & 3605 & 0.200 &\\
\hline
350 && 0.0197 & 192 & 0.163 && 0.0236 & 3612 & 0.144 && 0.0200 & 428 & 0.200 && 0.0194 & 3609 & 0.200 &\\
\hline
\end{tabular}
  \end{table}

\begin{table}\caption{The numerical results obtained when solving the mean-VaR-Basel problems with
real market data using the ADM and the MIP methods.}
  \label{tab:VaR-Basel-sp}
 \setlength{\tabcolsep}{1.5pt}
%\resizebox{\linewidth}{!}{
\centering
\begin{tabular}{ccccccccccccccccccccccccccccc}
\hline
\multicolumn{1}{c}{stocks} && \multicolumn{3}{c}{ $\rm ADM_{\BaselTwoDotFive\le C_0}$}&&  \multicolumn{3}{c}{ $\rm MIP_{\BaselTwoDotFive\le C_0}$} && \multicolumn{3}{c}{ $\rm ADM_{Basel3\le C_0}$}&&  \multicolumn{3}{c}{ $\rm MIP_{Basel3\le C_0}$} \\
\cline{3-5}\cline{7-9}\cline{11-13}\cline{15-17}
$d$  && $\rho_{\rm VaR}$ & time & $\rho_{\rm \BaselTwoDotFive}$ && $\rho_{\rm VaR}$ & time &  $\rho_{\rm \BaselTwoDotFive}$ && $\rho_{\rm VaR}$ & time & $\rho_{\rm Basel3}$ && $\rho_{\rm VaR}$ & time &  $\rho_{\rm Basel3}$\\
\hline
100 && 0.0246 & 148 & 0.138 && 0.0238 & 3602 & 0.143 && 0.0245 & 247 & 0.154 && 0.0237 & 3601 & 0.178 &\\
\hline
150 && 0.0249 & 167 & 0.146 && 0.0284 & 3602 & 0.160 && 0.0248 & 244 & 0.153 && 0.0259 & 3601 & 0.195 &\\
\hline
200 && 0.0228 & 174 & 0.129 && 0.0268 & 3602 & 0.159 && 0.0228 & 269 & 0.142 && 0.0222 & 3602 & 0.154 &\\
\hline
250 && 0.0233 & 194 & 0.134 && 0.0256 & 3606 & 0.153 && 0.0233 & 300 & 0.164 && 0.0248 & 3602 & 0.171 &\\
\hline
300 && 0.0224 & 193 & 0.133 && 0.0265 & 3606 & 0.158 && 0.0224 & 283 & 0.152 && 0.0291 & 3602 & 0.193 &\\
\hline
350 && 0.0227 & 208 & 0.131 && 0.0270 & 3609 & 0.146 && 0.0228 & 302 & 0.149 && 0.0293 & 3610 & 0.197 &\\
\hline
\end{tabular}
  \end{table}

\section{Conclusions}\label{sec:conclusion}
 A major change in financial regulations after the recent financial crisis is that financial institutions are now required to meet more stringent regulatory capital requirements than previously. It has been estimated that
the capital requirement for a large bank's trading book under the Basel 2.5 Accord on average
\emph{more than doubles} that under the Basel II Accord.  The significantly higher
capital requirement makes it more important for banks to take into account the
capital constraint when they construct their investment portfolios. In this paper, we propose a new asset allocation model, called the ``mean-$\rho$-Basel" model, that incorporates the Basel Accord capital requirements as one of the constraints. In this model, the capital requirement is measured using the Basel 2.5 and Basel III risk measures imposed by regulators; the risk level of the portfolio is measured by $\rho$, such as variance, VaR, and CVaR that can be freely chosen by the portfolio manager.

The complexity of the Basel 2.5 and Basel III risk measures, which involve risk
measurement under multiple scenarios, including stressed scenarios, poses
significant computational challenges to the proposed asset allocation problem due to its inherent non-convexity and non-smoothness.
We propose an unified algorithm based on the alternating
direction augmented Lagrangian method to solve the mean-$\rho$-Basel model and classical mean-$\rho$ model. The method is very simple and easy to implement; it
reduces the
original problem to one-dimensional optimization or convex quadratic programming subproblems
that may even have closed-form solutions; hence,
it is capable of solving large-scale problems that are difficult to solve using many other methods. For non-convex cases of the mean-$\rho$-Basel model, we  establish the first-order optimality of the limit points of the sequence generated by the method under some mild conditions.
%; for convex cases, the method is guaranteed to converge to globally optimal solutions.
Extensive  numerical results suggest that our method is promising for finding high-quality approximate optimal solutions, especially in non-convex cases.

% Appendix here
% Options are (1) APPENDIX (with or without general title) or
%             (2) APPENDICES (if it has more than one unrelated sections)
% Outcomment the appropriate case if necessary
%

\appendix

\section{Proof of Lemmas in Section \ref{subsubsec:algo_var}}
\subsection{Proof of Lemma \ref{lemma:opt-ADM-basel-x}}\label{app:proof-lemma-opt-ADM-basel-x}
\begin{proof}
%Let $v$ and $v^{[s]}$ be defined as in the lemma.
Since $\rho_{\BaselG}=\rho_{\BaselTwoDotFive}$,  problem \eqref{eq:ADM-x} is equivalent to
%\be \label{eq:ADM-x-v4}
%x^{(j+1)}  = \arg \min_{x}  \;
%  \| x - v\|^2, \text{ s.t. } \rho_{\BaselG}(x) \le C_0,
%\ee
\be \label{eq:ADM-x-v5}
\begin{aligned} \min_{x}  \;\; &  \phi(x):=\sum_{s=1}^m \| x^{[s]} - v^{[s]}\|^2 \\
 \text{ s.t. } \; \;&   \max\left\{ x^{[1]}_{(p_1)},
  \frac{k}{m_1}\sum_{s=1}^{m_1}x^{[s]} _{(p_s)}\right\}+\max\left\{
  x^{[m_1+1]}_{(p_{m_1+1})},
  \frac{\ell}{m_2}\sum_{s=m_1+1}^{m} x^{[s]}_{(p_s)}\right\} \le C_0.
  \end{aligned}
  \ee
Without loss of generality, assume that $v^{[s]}_1 \le v^{[s]}_2 \le \cdots \le v^{[s]}_{n_s}$, $s=1, \ldots, m$. Then, $(k_{s,1}, k_{s,2}, \ldots, k_{s,n_s})=(1, 2, \ldots, n_s)$. Let $x$ be an optimal solution to \eqref{eq:ADM-x-v5}.
If $x^{[s]}_i > x^{[s]}_j$ for some $i < j$, then since $v^{[s]}_i \leq v^{[s]}_j$,
switching the values of $x^{[s]}_i$ and $x^{[s]}_j$ will maintain the feasibility of $x$ without increasing $\phi(x)$. Thus, we can obtain an optimal solution $x$ that satisfies
$x^{[s]}_1 \le x^{[s]}_2 \le
 \cdots \le x^{[s]}_{n_s}$. In addition, it must hold that $x^{[s]}_i\leq v^{[s]}_i$ for all $i$; otherwise, if $x^{[s]}_i> v^{[s]}_i$ for some $i$, then setting $x^{[s]}_i=v^{[s]}_i$ will maintain the feasibility of $x$ but strictly reduce $\phi(x)$. Furthermore, it must hold that $x^{[s]}_j=v^{[s]}_j$ for all $j>p_s$; otherwise, if there is some $j>p_s$ such that $x^{[s]}_j<v^{[s]}_j$, setting $x^{[s]}_j=v^{[s]}_j$ will maintain the feasibility of $x$ but strictly reduce $\phi(x)$. Therefore, %Hence, we obtain
%that $y_{i} = w_i$, $i=p'+1,\ldots,n'$, and
problem \eqref{eq:ADM-x-v5} is equivalent to
\be \label{eq:ADM-x-v5-simplify}
\begin{aligned} \min_{x}  \;\; &  \phi(x)=\sum_{s=1}^m \| x^{[s]} - v^{[s]}\|^2 \\
 \text{ s.t. } \; \;&  x^{[s]}_{1}\leq x^{[s]}_{2}\leq \cdots\leq x^{[s]}_{p_s},\; s=1, \ldots, m,\\
 & x^{[s]}_{j}=v^{[s]}_{j},\; j=p_s+1, \ldots, n_s,\;s=1,\ldots,m,\\
 & \max\left\{ x^{[1]}_{p_1},
  \frac{k}{m_1}\sum_{s=1}^{m_1}x^{[s]} _{p_s}\right\}+\max\left\{
  x^{[m_1+1]}_{p_{m_1+1}},
  \frac{\ell}{m_2}\sum_{s=m_1+1}^{m} x^{[s]}_{p_s}\right\} \le C_0.
  \end{aligned}
  \ee
Hence, the optimal solution $x$ is given by \eqref{eq:solution-2}.
\end{proof}

\subsection{Proof of Lemma \ref{lemma:opt-ADM-basel3-x}}\label{app:proof-lemma-opt-ADM-basel3-x}
\begin{proof}
The problem \eqref{eq:ADM-x} with $\rho_{\BaselG}=\rho_{\BaselThree}$ is equivalent to
\be \label{eq:ADM-x-CVaR}
\begin{aligned} \min_{x}  \;\; &  \sum_{s=1}^m \| x^{[s]} - v^{[s]}\|^2 \\
 \text{ s.t. } \; \;&  \max\left\{\rho_{\CVaR}(x^{[m_1+1]}),
  \frac{\ell}{m_2}\sum_{s=m_1+1}^{m}\rho_{\CVaR}(x^{[s]})\right\}  \le C_0.
\end{aligned}
\ee
Let
$x=((x^{[1]})^{\top}, (x^{[2]})^{\top}, \ldots, (x^{[m]})^{\top})^{\top}$ be the optimal solution to \eqref{eq:ADM-x-CVaR}. Then apparently $x^{[s]}=v^{[s]}$, for $s=1, 2, \ldots, m_1$. By \eqref{equ:cvar_opt_rep},
\be\label{equ:cvar_rep}
\rho_{\CVaR}(x^{[s]})=\min_{t\in\mathbb{R}} \;\; t+\frac{1}{(1-\alpha)n_s}\sum_{i=1}^{n_s}(x^{[s]}_i-t)_+.
\ee
Then using \eqref{equ:cvar_rep}, it is easy to show that
$((x^{[m_1+1]})^{\top}, (x^{[m_1+2]})^{\top}, \ldots, (x^{[m]})^{\top})^{\top}$
is an optimal solution to \eqref{eq:ADM-x-CVaR-solu}, which completes the proof.
\end{proof}

\subsection{Proof of Lemma \ref{lemma:opt-ADM-y-variance}.}\label{sec:app_lemma_opt-ADM-y-variance}
\begin{proof}
The subproblem \eqref{eq:ADM-y} is equivalent to
\be\label{eq:ADM-y-v4}
 y^{(j+1)}  = \arg \min_{y }  \;
\rho(y) + \frac{\sigma_2}{2}  \| y - w^{(j)}\|^2,\ \text{where}\ w^{(j)} = - \left(\tilde Y u^{(j)} + \frac{1}{\sigma_2} \pi^{(j)} \right).
\ee
%where $w^{(j)} = - \left(\tilde Y u^{(j)} + \frac{1}{\sigma_2} \pi^{(j)} \right)$.
The result follows by using the definition of $\rho_{\Variance}$ given in
\eqref{eq:rho-var} and  Sherman--Morrison--Woodbury formula.
%\beaa
%y^{(j+1)}&=& \sigma_2
%\left(\left(\sigma_2+\frac{2}{n'}\right)I -\frac{2}{(n')^2}\one\one^{\top}\right)^{-1}
%w^{(j)} \\
%&=& \left(\sigma_2+\frac{2}{n'}\right)^{-1} \left( \sigma_2 w^{(j)} + 2\frac{\one^{\top} w^{(j)}}{(n')^2}
% \one  \right)
%\eeaa
\end{proof}

\subsection{Proof of Lemma \ref{lemma:opt-ADM-y-var}.}\label{sec:app_lemma_opt-ADM-y-var}
\begin{proof}
The problem \eqref{eq:ADM-y} with $\rho=\rho_{\VaR}$ becomes
\be  \min_{y}  \;
    \psi(y)=
y_{(p')} + \frac{\sigma_2}{2} \| y - w\|^2,\ \text{where}\ p':=\lceil \alpha n'\rceil. \label{equation:ADM-y-var}\ee
Without loss of generality, assume that $w_1 \le w_2 \le \ldots \le w_{n'}$. Then $(k_1, k_2, \ldots, k_{n'})=(1, 2, \ldots, n')$. Let $y$ be an optimal solution of \eqref{equation:ADM-y-var}.
%\comm{The following paragraph can be removed; it can be replaced by the paragraph in blue.}
%Let \begin{eqnarray}\label{eq:x-new3}
%q = \argmax\limits_{1\leq q\leq p'} y_q\; \mbox{ and } \; h =
%\argmin\limits_{h\geq p'+1\mbox{\ or\ }h=q} y_h.
%\end{eqnarray}
%Suppose for the sake of contradiction that $y_{q}$ is not the $p'$th smallest component of $y$, then
%\be \label{eq:x-new4}
% y_{q} > y_{h}\; \mbox{ and } \;
%w_{h} \ge  w_{p'}\ge  w_{q}. \ee
%We will show that there exists $z$ such that $\psi(y)>\psi(z)$, which leads to a contradiction. In fact, there are two cases: (i) If $y_{q} > w_{q}$, then we define $z$ as $z_{q} = w_{q}$ and $z_{j}=y_{j}$ for $\forall j\neq q$. Apparently $\psi(y)>\psi(z)$. (ii) If $y_{q} \le w_{q}$, then we define $z$ as $z_{q} = y_{h}$, $z _{h} = w _{h}$ and  $z_{j}=y_{j}$ for $\forall j\neq q, h$. It follows from
%from \eqref{eq:x-new4} that
%$w_{h} \geq w_{q} \geq  y_{q} > y _{h}$; hence, $\psi(y)>\psi(z)$. Therefore, $y_{q}$ must be the $p'$th smallest component of $y$, and $y_{q} \leq y_{h}$.
If $y_i > y_j$ for some $i < j$, then since $w_i\leq w_j$, switching the values of $y_i$ and $y_j$ will not increase $\psi(y)$. Thus, we can obtain an optimal solution $y$ that satisfies $y_1 \le y_2 \le
 \cdots \le y_{n'}$. In addition, the optimal solution $y$ must satisfy that $y_i\leq w_i$ for all $i$; otherwise, if $y_i>w_i$ for some $i$, then setting $y_i=w_i$ will strictly reduce $\psi(y)$. Furthermore, it must hold that $y_j=w_j$ for all $j=p'+1, p'+2, \ldots, n'$; otherwise, if there is some $j>p'$ such that $y_j<w_j$, setting $y_j=w_j$ will strictly reduce $\psi(y)$. Therefore, %Hence, we obtain
%that $y_{i} = w_i$, $i=p'+1,\ldots,n'$, and
 problem \eqref{equation:ADM-y-var} is equivalent to
%$(y_1, \cdots, y_{p'})$ is the optimal
%solutions to
\be \label{eq:ADM-x-v3}
\begin{aligned}  \min_{y}  \; \; &
  y_{p'} + \frac{\sigma_2}{2} \sum_{i=1}^{p'} (y_i - w_i)^2\\
  \text{s.t. }\;\; & y_i \le y_{p'}, \; i=1, \cdots, p'-1,\\
%  & \revise{y_1\leq y_2\leq \cdots \leq y_{p'},}\\
  & y_j=w_j,\; j=p'+1, p'+2, \ldots, n'.
 \end{aligned}
 \ee
The KKT conditions of \eqref{eq:ADM-x-v3} are
\bea
\label{eq:KKT-1}
%y_1\leq y_2\leq \cdots & \leq & y_{p'},\\
y_i & \le & y_{p'}, \; i=1, \ldots, p'-1, \\
\label{eq:KKT-2}
\sigma_2 (y_i-w_i) + \bar \pi_i &=&0,\; i=1, \ldots, p'-1, \\
\label{eq:KKT-2dot5}
\bar \pi_i (y_{p'} - y_i) &=& 0,\; i=1, \ldots, p'-1,\\
\label{eq:KKT-3}
\sigma_2 (y_{p'} -w_{p'}) + 1 -  \sum_{j=1}^{p'-1} \bar \pi_j &=& 0,\\
\label{eq:KKT-4}
\bar \pi_i&\geq& 0,\; i=1, \ldots, p'-1.
\eea
Since problem \eqref{eq:ADM-x-v3} is convex, the KKT conditions are also sufficient for the optimality of $y$.
The equations \eqref{eq:KKT-2} and \eqref{eq:KKT-2dot5} imply that for each $i=1,\ldots,p'-1$, either $y_i=w_i$ (if $\bar \pi_i=0$) or $y_i=y_{p'}$ (if $\bar \pi_i>0$). Since $y_1\leq \cdots \leq y_{p'}$, it follows that there exists $1\leq i^*\leq p'$ such that $y_j=w_j$ for $j<i^*$, $y_j=y_{p'}$ for $j\geq i^*$, and
$w_{i^*-1}<y_{p'}$. Then by \eqref{eq:KKT-3}, we have $y_{p'} = \gamma_{i^*}$, where $\gamma_i$ is defined in \eqref{equ:opt_x_order-v2}. It follows from \eqref{eq:KKT-2} and \eqref{eq:KKT-4} that $y_j\leq w_j$, $j=1, \ldots, p'$. Hence, we have $\gamma_{i^*}=y_{i^*}\leq w_{i^*}$. Therefore, $i^*$ should satisfy $w_{i^*-1}<\gamma_{i^*}\leq w_{i^*}$, which completes the proof.
\end{proof}

\subsection{Proof of Lemma \ref{lemma:opt-ADM-y-cvar}}\label{app:proof-lemma-opt-ADM-y-cvar}
\begin{proof}
With $\rho=\rho_{\CVaR}$, it follows from \eqref{equ:cvar_rep} that problem \eqref{eq:ADM-y} is equivalent to
  %There even exists closed-form solutions in the VaR case.
 %In addition, we let $v_0=-\infty$.
\be  \min_{t,y}  \;
    \phi(t,y)=
    t + \frac{1}{(1-\alpha)n'}\sum_{i=1}^{n'} (y_i-t)_+ + \frac{\sigma_2}{2} \| y - w^{(j)}\|^2,\label{equ:xh-CVaR}
\ee
where $x_+:=\max(x, 0)$. For any fixed $t$, the optimal $y$ that
minimizes $\phi(t, y)$ is $y(t)$ defined in \eqref{equ:xh-CVaR-yt}. Hence, the result follows. \end{proof}

 \section{Proof of Proposition \ref{lm:1}.}\label{app:proof_lm1}
\begin{proof}
We first show that for any fixed $1\leq p\leq n$, $f(x):=x_{(p)}$ is globally Lipschitz.
For any given $x\in\mathbb{R}^n$, define
%let $L_{x_{(p)}}$, $E_{x_{(p)}}$ and
%$G_{x_{(p)}}$ be the set of all indices whose corresponding entries
%of $x$ are strictly less than, equal to, and strictly greater than $x_{(p)}$,
%respectively; that is,
\begin{equation}\label{equ:LEG}
L_{x_{(p)}}:= \{i\mid
x_i<x_{(p)}\},\ E_{x_{(p)}}:=\{i\mid x_i=x_{(p)}\}, \ G_{x_{(p)}}:= \{i\mid x_i>x_{(p)}\}.
\end{equation}
%\end{definition}
%It is quite clear that \begin{eqnarray} |\{i\mid x_i < x_{(p)}\}| \leq p-1\,\mbox{and}\,|\{i\mid x_i > x_{(p)}\}| \leq n-p.\end{eqnarray}
%Given $\forall\,x\in\mathbb{R}^n$,
%There are two cases:
%
%Case (i): $L_{x_{(p)}}\cup G_{x_{(p)}}\neq \emptyset$. Then $\delta_x :=
%\min\limits_{i\in L_{x_{(p)}}\cup G_{x_{(p)}}} |x_i - x_{(p)} |>0$.  Then for any
%$y,z\in B(x,\frac{\delta_x}{6})$, it can be  proved that
%$\delta_z \geq \frac{2\delta_x}{3}$ and $y\in
%B(z,\frac{\delta_x}{3}) \subseteq B(z,\frac{\delta_z}{2}) $. Hence,
%we  have
%\begin{eqnarray}\label{eq:tobechecked0410-1}
%y_i < y_j,\quad \forall\, i\in L_{z_{(p)}},\,j\in E_{z_{(p)}};\qquad
%y_i > y_j,\quad \forall\, i\in G_{z_{(p)}},\,j\in E_{z_{(p)}}.
%\end{eqnarray}
%It follows from \eqref{eq:tobechecked0410-1} that
%$E_{y_{(p)}}\subset E_{z_{(p)}}$ and consequently
%\begin{eqnarray}\label{eq:new42-1}
%y_{(p)}-z_{(p)} = y_i - z_j = y_i - z_i + z_i - z_j = y_i -
%z_i,\quad \forall\,i\in E_{y_{(p)}},\,\forall\,j\in E_{z_{(p)}}.
%\end{eqnarray}
%Furthermore, $|y_{(p)}-z_{(p)}| = |y_i - z_i| \leq \|y-z\|$, which
%means that $f(x)=x_{(p)}$ is Lipschitz near $x$. %Since $x$ is arbitrary, the local Lipschitzness of $\rho_{\VaR}(x)$ is proved.
%
%Case (ii): $L_{x_{(p)}}\cup G_{x_{(p)}}=\emptyset$, namely, all the
%components of $x$
%take the same value.
For any given $y, z\in\mathbb{R}^n$, without loss of generality, assume that $y_{(p)} \leq z_{(p)}$. It follows from the definition of $L_{y_{(p)}}$ and $E_{y_{(p)}}$ that the number of elements of $L_{y_{(p)}}\cup E_{y_{(p)}}$ is strictly larger than that of $L_{z_{(p)}}$. Therefore, the set $I:=(L_{y_{(p)}}\cup E_{y_{(p)}})\cap (E_{z_{(p)}}\cup G_{z_{(p)}})$ is not empty. Choose any $i\in I$. Then $y_i\leq y_{(p)}$ and $z_i\geq z_{(p)}$. Hence, $|y_{(p)} - z_{(p)}|=z_{(p)}-y_{(p)}\leq z_i - y_i \leq \|y-z\|$, which establishes that $\rho_{\VaR}(x)$ is globally Lipschitz.

Using the inequality $|\max(a, b)-\max(c, d)|\leq |a-c|+|b-d|$ for $\forall a, b, c, d\in\mathbb{R}$, it can be shown that the maximum of two globally  Lipschitz functions is also globally Lipschitz. Since $\rho_{\CVaR}$, $\rho_{\BaselTwoDotFive}$, and $\rho_{\BaselThree}$ are all finite linear combination of (maximum of) globally
Lipschitz functions, it follows that they are all globally Lipschitz.
%Note that $\rho_{\BaselTwoDotFive}(x)$ %and $\rho_{\TTCE}(x)$ are
%is a composite function of finite linear combinations of $\rho_{\VaR}(x)$ (with different $p$) and
% maximization. Since the maximization of two
%locally  Lipschitz functions and the summation of a finite number of locally
%Lipschitz functions are still
%locally Lipschitz, we can prove that
%$\rho_{\BaselTwoDotFive}(x)$
%%and $\rho_{\TTCE}(x)$ are
%is locally  Lipschitz.
\end{proof}

%\begin{proof}
%\revise{We first show that for any fixed $1\leq p\leq n$, $f(x):=x_{(p)}$ is locally Lipschitz. For any given $x, y\in\mathbb{R}^n$, let $(i_1, i_2, \ldots, i_n)$ and $(j_1, j_2, \ldots, j_n)$ be two permutations of $(1, 2, \ldots, n)$ such that $x_{i_1}\leq x_{i_2}\leq \cdots \leq x_{i_n}$ and $y_{j_1}\leq y_{j_2}\leq \cdots \leq y_{j_n}$. Then
%  \begin{equation}\label{equ:norm}
%  |x_{(p)}-y_{(p)}|=|x_{i_p}-y_{j_p}|\leq (\sum_{k=1}^n (x_{i_k}-y_{j_k})^2)^{1/2}.
%  \end{equation}
%Let $P_y:=\{(y_{l_1}, y_{l_2}, \ldots, y_{l_n})\mid (l_1, l_2, \ldots, l_n)\ \text{is a permutation of}\ (1, 2, \ldots, n)\}$ and define the function
%  $H(z):=(\sum_{k=1}^n (x_{i_k}-z_{k})^2)^{1/2}, z\in P_y.$
%    For any $k_1<k_2$, if $z_{k_1}>z_{k_2}$, then, since $x_{i_{k_1}}\leq x_{i_{k_2}}$, swapping the values of $z_{k_1}$ and $z_{k_2}$ reduces $H(z)$. Hence, $H(z)$ is minimized at $(y_{j_1}, y_{j_2}, \cdots, y_{j_n})$, which in combination with \eqref{equ:norm} implies that $|x_{(p)}-y_{(p)}|\leq H((y_{j_1}, y_{j_2}, \ldots, y_{j_n}))\leq H((y_{i_1}, y_{i_2}, \ldots, y_{i_n}))=\|x-y\|$.}
%\end{proof}

\section{Proof of Theorem \ref{thm:cvg-feasi}.}\label{app:proof-theorem-thm:cvg-feasi}
Let $\mathrm{conv}(A)$ denote the convex hull of $A$. First, we prove the following two propositions.
\begin{proposition}\label{prop:clarke_gradient_var}
Let $e_i$ be the $i$th standard basis vector in $\R^n$. The Clarke's generalized gradient of $f(x)=x_{(p)}$ is given by
\begin{align}\label{equ:clarde_gradient_var}
\bar \partial x_{(p)}=\mathrm{conv}\{e_i\mid i\in
E_{x_{(p)}}\},\ \text{where}\ E_{x_{(p)}} := \{i\mid x_i=x_{(p)}\}.
\end{align}
\end{proposition}
\begin{proof}
For any $x\in\mathbb{R}^n$ and $d\in\mathbb{R}^n$, let $f^\circ(x;d)$ be the
Clarke's generalized directional derivative at $x$ along the direction $d$, i.e.,
\begin{eqnarray}\label{eq:tobecheckednew2}
f^\circ(x;d) := \limsup_{y\rightarrow
x,\,t\rightarrow0^+} \frac{f(y+td)-f(y)}{t}.
\end{eqnarray}
Define $d_{\max}(z):=\max\{d_i\mid i\in
E_{z_{(p)}}\}, z\in\mathbb{R}^n$. First, we will show that
\begin{equation}\label{equ:cgd}
f^\circ(x;d)=d_{\max}(x).
\end{equation}
Indeed, suppose that $L_{x_{(p)}}=\{i_{1}, \ldots, i_{k}\}$ and $E_{x_{(p)}}=\{i_{k+1}, \ldots, i_{k+l}\}$. Then, $k+1\leq p\leq k+l$.
By the definitions in \eqref{equ:LEG}, there exists $\eta>0$ such that for any $(y, t)\in B(x, \eta)\times (0, \eta)$ it holds that
$(y+td)_{i}<(y+td)_{j}<(y+td)_{k}\ \text{and}\ y_{i}<y_{j}<y_{k},\ \text{for}\ \forall i\in L_{x_{(p)}}, \forall j\in E_{x_{(p)}}, \forall k\in G_{x_{(p)}}$.
For any such $(y, t)$, $y_{(p)}=(y_{i_{k+1}}, y_{i_{k+2}}, \ldots, y_{i_{k+l}})_{(p-k)}$ and $(y+td)_{(p)}=((y+td)_{i_{k+1}}, (y+td)_{i_{k+2}}, \ldots, (y+td)_{i_{k+l}})_{(p-k)}$.
Suppose without loss of generality that $y_{i_{k+1}}\leq y_{i_{k+2}}\leq \cdots \leq y_{i_{k+l}}$. Then $y_{(p)}=y_{i_p}$.
Let $j'\leq p-k$ be the index such that
$(y+td)_{i_{k+j'}}=\max\{(y+td)_{i_{k+1}}, (y+td)_{i_{k+2}}, \ldots, (y+td)_{i_{p}}\}$. Then $(y+td)_{(p)}=((y+td)_{i_{k+1}}, (y+td)_{i_{k+2}}, \ldots, (y+td)_{i_{k+l}})_{(p-k)}\leq \max\{(y+td)_{i_{k+1}}, (y+td)_{i_{k+2}}, \ldots, (y+td)_{i_{p}}\}=(y+td)_{i_{k+j'}}$. Furthermore, $j'\leq p-k$ implies that
$y_{i_p}\geq y_{i_{k+j'}}$. Therefore,
\begin{equation}\label{equ:c2}
\frac{f(y+td)-f(y)}{t}=\frac{(y+td)_{(p)}-y_{i_p}}{t}\leq \frac{(y+td)_{i_{k+j'}}-y_{i_{k+j'}}}{t}= d_{i_{k+j'}}\leq d_{\max}(x).
\end{equation}
Since \eqref{equ:c2} holds for any $(y, t)\in B(x, \eta)\times (0, \eta)$, it follows that
\begin{equation}\label{eq:tobechecked0410-2}
  f^\circ(x;d)\leq d_{\max}(x).
\end{equation}
On the other hand, suppose $d_{i_{k+j^*}}=d_{\max}(x)$. Define $\zeta:=1+\max\{|d_i|\mid i\in E_{x_{(p)}}\}$. There exists a sequence $y^{(m)}\to x$ as $m\to \infty$ such that for all $m$ it holds that $y^{(m)}_{(p)}=(y^{(m)}_{i_{k+1}}, y^{(m)}_{i_{k+2}}, \ldots, y^{(m)}_{i_{k+l}})_{(p-k)}=y^{(m)}_{i_{k+j^*}}$ and $\min\{|y^{(m)}_{i_{k+a}}-y^{(m)}_{i_{k+b}}|\mid a\neq b, 1\leq a, b\leq l\}=2^{-m}\zeta$. Define $t^{(m)}:=2^{-m-2}$. Then
\begin{equation}\label{equ:seq_lim}
  \frac{f(y^{(m)}+t^{(m)}d)-f(y^{(m)})}{t^{(m)}}=d_{i_{k+ j^*}}=d_{\max}(x),\ \forall m.
\end{equation}
Combining \eqref{equ:seq_lim} with \eqref{eq:tobechecked0410-2}, we obtain \eqref{equ:cgd}.

Second, we will show that \eqref{equ:clarde_gradient_var} holds. By definition, $\bar \partial f(x) := \{\xi\in\mathbb{R}^n\mid \xi^\top d \leq f^\circ(x;d),\,\forall d\in
\mathbb{R}^n\}$.
On one hand, for $\forall \xi \in{\rm conv}\{e_i\mid i\in
E_{x_{(p)}}\}$, $\xi$ can be represented by $\xi = \sum_{i\in
E_{x_{(p)}}}c_i e_i$, where $c_i\geq 0$ for all $i\in
E_{x_{(p)}}$ and $\sum_{i\in E_{x_{(p)}}}c_i = 1$. Hence,
$\xi^\top d \leq d_{\max}(x)=f^\circ(x;d)$ for $\forall\,d\in \mathbb{R}^n$, which implies that ${\rm conv}\{e_i\mid i\in E_{x_{(p)}}\}\subseteq \bar\partial x_{(p)}$. On the other hand, for $\forall\, \xi\notin{\rm
conv}\{e_i\mid i\in E_{x_{(p)}}\}$, it follows from separating hyperplane theorem that there exists $d\in
\mathbb{R}^n$ and $\alpha\in \mathbb{R}$ such that
$\xi^\top d > \alpha \geq \sup_{\mu\in{\rm
conv}\{e_i\mid i\in E_{x_{(p)}}\}}\mu^\top d =
d_{\max}(x)=f^\circ(x;d)$, which implies $\xi\notin \bar\partial x_{(p)}$. Therefore, $\bar\partial x_{(p)}\subseteq {\rm conv}\{e_i\mid i\in E_{x_{(p)}}\}$. Hence, \eqref{equ:clarde_gradient_var} follows.
%i.e. $ \partial x_{(p)}
%\subset{\rm conv}\left\{e_i\mid i\in E_{x_{(p)}}\right\}$. Together
%with \eqref{eq:tobechecked0410-3}, \eqref{equ:clarde_gradient_var} holds and
%we complete the proof.
\end{proof}
\begin{proposition}\label{prop:bounded_away_zero}
For $\rho_{\BaselG}\in\{\rho_{\BaselTwoDotFive}, \rho_{\BaselThree}\}$, there exists a closed and bounded set $\mathcal{C}\subset \mathbb{R}^n_+$ such that $0\notin \mathcal{C}$ and $\bar \partial \rho_{\BaselG}(x)\subset \mathcal{C}$ for any $x\in \mathbb{R}^n$.
%there exists $\delta>0$ such that for any $x\in\mathbb{R}^n$ and any $y\in \bar \partial \rho_{\BaselG}(x)$ it holds that $\|y\|^2\geq \delta$.
\end{proposition}
\begin{proof}
%The Clarke's Generalized Gradients of $\rho_{\VaR}(x)$ and $\rho_{\BaselTwoDotFive}(x)$
% = {\rm conv}\left\{\bigcup\limits_{i\in E_{x_{(p)}}} \{e_i\}\right\},
%$e_i$ is the vector with a one on the $i$th position and zeros everywhere else.
Let $e^{[s]}_i$ be the $i$th standard basis in $\mathbb{R}^{n_s}$ and
$E_{x^{[s]}_{(p)}} := \{1\leq i\leq n_s\mid x^{[s]}_i=x^{[s]}_{(p)}\}$. By similar argument in the proof of Proposition \ref{prop:clarke_gradient_var}, it can be shown that
\begin{equation}\label{equ:gpartial_var_s}
\bar \partial x^{[s]}_{(p)}=\mathrm{conv}\{(0,\ldots, 0, e^{[s]}_i, 0, \ldots, 0)\mid i\in
E_{x^{[s]}_{(p)}}\}.
\end{equation}
Then by Theorem 2.3.3 and Theorem 2.3.10 in \citet{Clarke1990}, we have
\begin{eqnarray}\label{eq:clarkeggrad2}
\bar \partial \rho_{\BaselTwoDotFive}(x)\subseteq \mathrm{conv}\left(\cup_{i\in I_1(x)}\bar \partial f_i(x)\right)  +
\mathrm{conv}\left(\cup_{i\in I_2(x)}\bar \partial f_i(x)\right),
\end{eqnarray}
where
%\begin{description}
%\item[(1)]
$f_1(x) = x^{[1]}_{(p_1)}$, and $\bar \partial f_1(x)= \mathrm{conv}\{(e^{[1]}_i, 0, \ldots, 0)\mid i\in
E_{x^{[1]}_{(p_1)}}\}$;
%\item[(2)]
$f_2(x) = \frac{k}{m_1}\sum_{s=1}^{m_1}x^{[s]}_{(p_s)}$, and
$\bar \partial f_2(x)\subseteq
\frac{k}{m_1}\sum_{s=1}^{m_1}\mathrm{conv}\{(0, \ldots, 0, e^{[s]}_i, 0, \ldots, 0)\mid i\in
E_{x^{[s]}_{(p_s)}}\}$;
%\item[(3)]
$I_1(x):=\{i\mid \max\{f_1(x),\linebreak f_2(x)\}=f_i(x),\; i\in\{1,2\}\}$;
%\item[(4)]
$f_3(x) = x^{[m_1+1]}_{(p_{m_1+1})}$, and $\bar \partial f_3(x)= \mathrm{conv}\{(0,\ldots, 0, e^{[m_1+1]}_i, 0, \ldots, 0)\mid i\in
E_{x^{[m_1+1]}_{(p_{m_1+1})}}\}$;
%\item[(5)]
$f_4(x) = \frac{l}{m_2}\sum_{s=m_1+1}^{m}x^{[s]}_{(p_s)}$, and
$\bar \partial f_4(x)\subseteq
\frac{l}{m_2}\sum_{s=m_1+1}^{m}\mathrm{conv}\{(0, \ldots, 0, e^{[s]}_i, 0,\linebreak \ldots, 0)\mid
i\in E_{x^{[s]}_{(p_s)}}\}$;
%\item[(6)]
$I_2(x):=\{i\mid \max\{f_3(x),f_4(x)\}=f_i(x),\; i\in\{3,4\}\}$.
%\end{description}
Define $A_1:= \mathrm{conv}\{(e^{[1]}_i, 0, \ldots, 0)\mid 1\leq i\leq n_1\}$,
$A_2: =
\frac{k}{m_1}\sum_{s=1}^{m_1}\mathrm{conv}\{(0, \ldots, 0, e^{[s]}_i, 0, \ldots, 0)\mid 1\leq i\leq n_s\}$,
$A_3:= \mathrm{conv}\{(0,\ldots, 0, e^{[m_1+1]}_i, 0, \ldots, 0)\mid 1\leq i\leq n_{m_1+1}\}$, $A_4 := \frac{l}{m_2}\sum_{s=m_1+1}^{m}\mathrm{conv}\{(0, \ldots, 0, e^{[s]}_i, 0, \ldots, 0)\mid
1\leq i\leq n_s\}$, and $\mathcal{C}:=\mathrm{conv}\left(A_1\cup A_2\right)  +
\mathrm{conv}\left(A_3\cup A_4\right)$. Then, $\mathcal{C}$ is compact and $0\notin\mathcal{C}$; in addition, it follows from
\eqref{eq:clarkeggrad2} that $\bar\partial \rho_{\BaselTwoDotFive}(x)\subset \mathcal{C}$ for any $x\in\mathbb{R}^n$.
%it is easy to see that there exists such a set $\mathcal{C}$ such that the conclusion of the proposition holds for $\rho_{\BaselTwoDotFive}$.
%there exists a bounded closed set $\mathcal{C}\subset \mathbb{R}^n_+$ such that $0\notin \mathcal{C}$ and $\bar \partial \rho_{\BaselTwoDotFive}(x)\subseteq \mathcal{C}$ for all $x\in \mathbb{R}^n$. Hence, $\min_{y\in \bar \partial \rho_{\BaselTwoDotFive}(x), x\in\mathbb{R}^n}\|y\|^2\geq \min_{y\in \mathcal{C}}\|y\|^2>0$.

By \eqref{equ:e_s_cvar} and \eqref{equ:gpartial_var_s}, we have
\begin{align*}
  \bar\partial \rho_{\CVaR}(x^{[s]})\subseteq{}& \frac{p_s-\alpha n_s}{(1-\alpha)n_s}\mathrm{conv}\{(0,\ldots, 0, e^{[s]}_i, 0, \ldots, 0)\mid i\in
E_{x^{[s]}_{(p_s)}}\}\\
{}&+\frac{1}{(1-\alpha)n_s}\sum_{j=p_s+1}^{n_s}\mathrm{conv}\{(0,\ldots, 0, e^{[s]}_i, 0, \ldots, 0)\mid i\in
E_{x^{[s]}_{(j)}}\}.
\end{align*}
Then, we can show by similar argument that the conclusion also holds for $\rho_{\BaselG}=\rho_{\BaselThree}$.
\end{proof}

The proof of Theorem \ref{thm:cvg-feasi} is as follows.

%\begin{definition}\label{def:2} Let $f(x): \R^n \to \R$. The Clarke's
%  generalized gradient of $f(x)$ at $x\in \mathrm{dom}f$ is defined as
%  \begin{eqnarray}\label{eq:tobecheckednew1}
%\bar \partial f(x) := \{\xi\in\mathbb{R}^n\mid \xi^\top d \leq f^\circ(x;d),\,\forall d\in
%\mathbb{R}^n\},
%\end{eqnarray}
%where, for any $d\in \R^n$, $f^\circ(x;d)$ is the Clarke's generalized directional derivative
%\begin{eqnarray}\label{eq:tobecheckednew2}
%f^\circ(x;d) := \limsup\limits_{y\rightarrow
%x,\,t\downarrow 0} \frac{f(y+td)-f(y)}{t}.
%\end{eqnarray}
%\end{definition}

\begin{proof}[Proof of Theorem \ref{thm:cvg-feasi}.] %\begin{proof}
(i) Since $\rho(x)$ is locally Lipschitz and $\rho_{\BaselG}(x)$ is globally Lipschitz on $\mathbb{R}^n$ (by Lemma \ref{lm:1}), it follows from Proposition 2.1.2 in \cite{Clarke1990} that $\bar \partial \rho(x)$ and $\bar \partial \rho_{\BaselG}(x)$ exist on $\R^{n}$. Then the first part of the theorem follows from the corollary of Proposition 2.4.3 and Theorem 2.3.10 in \citet{Clarke1990}.

(ii) To prove part (ii), we first show that
  \be \label{eq:cvg-w} \lim_{j\to\infty}x^{(j+1)}-x^{(j)} = 0, \; \lim_{j\to\infty}y^{(j+1)}-y^{(j)} = 0,  \mbox{ and } \lim_{j\to\infty}u^{(j+1)}-u^{(j)} =0.\ee
  Since $\Ur$ is closed and bounded, the sequence $\{u^{(j)}\}$
 is bounded.
It follows from \eqref{eq:ADM-lmb} and \eqref{eq:ADM-pi} that  $x^{(j+1)}=(\lambda^{(j+1)}-\lambda^{(j)})/(\beta_1\sigma_1)-\tR u^{(j+1)}$ and $y^{(j+1)}=(\pi^{(j+1)}-\pi^{(j)})/(\beta_2\sigma_2)-\tY u^{(j+1)}$, which in combination with boundedness of $\{u^{(j)}\}$ and assumed boundedness of $\{(\lambda^{(j)}, \pi^{(j)})\}$ implies that $\{(x^{(j)},y^{(j)})\}$ is bounded. Thus, $\{(x^{(j)}, y^{(j)}, u^{(j)}, \lambda^{(j)}, \pi^{(j)})\}$ is bounded, and then the continuity of the augmented Lagrangian
function \eqref{auglang} implies that
$\{\Lc(x^{(j)},y^{(j)},u^{(j)}, \lambda^{(j)}, \pi^{(j)})\}$ is
bounded.

Note that the  augmented Lagrangian
function $\Lc$ is strongly convex with respect to the variable $u$. Therefore, it holds
 that for any $u$ and $\Delta u$,
 \be \label{eq:convex-u1} \Lc(x,y,u+\Delta u, \lambda,\pi) - \Lc(x, y,u,
 \lambda,\pi) \ge
 \partial_u \Lc(x,y,u,\lambda,\pi)^\top
 \Delta u + c \|\Delta u \|^2,  \ee
where $c>0$ is constant. In addition,
since $u^{(j+1)}$ minimizes \eqref{eq:ADM-u}
and $u^{(j)} \in \Ur$, it follows that
\be\label{eq:convex-u2}  \partial_u \Lc(x^{(j+1)}, y^{(j+1)},u^{(j+1)},\lambda^{(j)}, \pi^{(j)})^\top (u^{(j)} - u^{(j+1)}) \ge 0.\ee
Combining \eqref{eq:convex-u1} and \eqref{eq:convex-u2}, we obtain
\be \label{eq:convex-u3}  \Lc(x^{(j+1)},y^{(j+1)}, u^{(j)}, \lambda^{(j)}, \pi^{(j)}) - \Lc(x^{(j+1)}, y^{(j+1)},u^{(j+1)}, \lambda^{(j)}, \pi^{(j)}) \ge c \|u^{(j+1)}-u^{(j)}\|^2.
 \ee
Since $x^{(j+1)}$ minimizes
 \eqref{eq:ADM-x} and $y^{(j+1)}$ minimizes
 \eqref{eq:ADM-y}, it follows from \eqref{eq:ADM-lmb}, \eqref{eq:ADM-pi}, and \eqref{eq:convex-u3} that
\bea \label{eq:red-L}
& & \Lc(x^{(j)},y^{(j)},u^{(j)}, \lambda^{(j)}, \pi^{(j)}) -
\Lc(x^{(j+1)},y^{(j+1)}, u^{(j+1)},\lambda^{(j+1)}, \pi^{(j+1)}) \nonumber \\
%&=&\Lc(x^{(j)},y^{(j)},u^{(j)}, \lambda^{(j)}, \pi^{(j)})
%-\Lc(x^{(j+1)},y^{(j)},u^{(j)}, \lambda^{(j)}, \pi^{(j)}) \nonumber \\
%& & + \Lc(x^{(j+1)},y^{(j)},u^{(j)}, \lambda^{(j)}, \pi^{(j)}) -
%\Lc(x^{(j+1)},y^{(j+1)},u^{(j)}, \lambda^{(j)}, \pi^{(j)}) \nonumber \\
%& & + \Lc(x^{(j+1)},y^{(j+1)},u^{(j)}, \lambda^{(j)}, \pi^{(j)}) -
%\Lc(x^{(j+1)},y^{(j+1)},u^{(j+1)}, \lambda^{(j)}, \pi^{(j)}) \nonumber \\
%& & + \Lc(x^{(j+1)},y^{(j+1)},u^{(j+1)}, \lambda^{(j)}, \pi^{(j)}) -
%\Lc(x^{(j+1)},y^{(j+1)},u^{(j+1)}, \lambda^{(j+1)}, \pi^{(j+1)})
%\nonumber\\
& & + \frac{1}{\beta_1 \sigma_1} \| \lambda^{(j)} -
\lambda^{(j+1)}\|^2+  \frac{1}{\beta_2 \sigma_2} \| \pi^{(j)} -
\pi^{(j+1)}\|^2 \ge c \|u^{(j+1)}-u^{(j)}\|^2.
\eea
Since $\sum_{j=1}^{\infty}   (\|\lambda^{(j+1)} - \lambda^{(j)}\|^2 + \|\pi^{(j+1)} -
  \pi^{(j)}\|^2)<\infty$ and $\{\Lc(x^{(j)},y^{(j)},u^{(j)},
\lambda^{(j)}, \pi^{(j)}) \}$ is bounded, it follows from
\eqref{eq:red-L} that
\be \sum_{j=1}^\infty \|u^{(j+1)}-u^{(j)}\|^2 < \infty, \ee
which implies that
\be\label{eq:convg-u}  \lim_{j\to \infty} u^{(j+1)}-u^{(j)} = 0. \ee
Since $\sum_{j=1}^{\infty}   (\|\lambda^{(j+1)} - \lambda^{(j)}\|^2 + \|\pi^{(j+1)} -
  \pi^{(j)}\|^2)<\infty$, it follows that
$\lim_{j\to \infty} \lambda^{(j+1)} - \lambda^{(j)} = 0$, which in combination with
\eqref{eq:ADM-lmb} implies that
\be
\label{eq:prob-feasi-cvg} \lim_{j\to \infty}  x^{(j+1)} + \tR
u^{(j+1)} = 0.\ee
By \eqref{eq:convg-u} and \eqref{eq:prob-feasi-cvg}, we obtain
$\lim_{j\to \infty} x^{(j+1)}-x^{(j)} = 0$. By similar argument, we obtain $\lim_{j\to \infty}   y^{(j+1)}-y^{(j)} = 0$.
% which together with Lemma \ref{lemma:opt-ADM-y-var} gives $\gamma^{(j+1)}-\gamma^{(j)} \to 0 $.
%\be\label{eq:convg-x}  \gamma^{(j+1)}-\gamma^{(j)} \to 0. \ee

For any limit point $\bar u$ of the sequence $\{u^{(j)}\}$, there exists a subsequence $u^{(k_i)}\to \bar u$ as $i\to\infty$. Since $\{(x^{(j)},y^{(j)},u^{(j)}, \lambda^{(j)}, \pi^{(j)})\}$ is bounded, there exists a further subsequence $\{j_i\}\subseteq \{k_i\}$ such that $(x^{(j_i)},y^{(j_i)}, u^{(j_i)}, \lambda^{(j_i)}, \pi^{(j_i)})\to(\bar x, \bar y, \bar u, \bar \lambda, \bar \pi)$ as $i\to\infty$.
%If $(\bar x, \bar y, \bar u, \bar \lambda, \bar \pi)$
%is a limit point
%of the sequence
%$\{(x^{(j)},y^{(j)},u^{(j)}, \lambda^{(j)}, \pi^{(j)})\}$, then there exists a
%subsequence .
Clearly, we obtain from
\eqref{eq:convg-u}  and
\eqref{eq:prob-feasi-cvg} that
\be \label{eq:cvg-con2}
\bar x
+ \tR \bar u = \lim_{i \to \infty} x^{(j_i)} + \tR u^{(j_i)} = \lim_{i \to \infty}
x^{(j_i)} + \tR u^{(j_i-1)} = 0.  \ee
A similar argument leads to $ \bar y + \tY \bar u = 0$.

The first-order optimality condition of \eqref{eq:ADM-x} in the
$j_i$th iteration is
\bea
0 \in   \sigma_1 (x^{(j_i)} + \tR u^{(j_i-1)}) +
\lambda^{(j_i-1)}   + \eta^{(j_i)} \bar \partial \rho_{\BaselG}(x^{(j_i)}), \label{eq:ADM-opt-x}\\
\eta^{(j_i)} (\rho_{\BaselG}(x^{(j_i)}) - C_0) = 0,\label{equ:adm-eta}
\eea
for some %$\zeta \in \{0,1\} $ and
$\eta^{(j_i)}\ge 0$. Since $\{\bar \partial \rho_{\BaselG}(x^{(j_i)}) \}$ is
bounded away from zero (by Proposition \ref{prop:bounded_away_zero}) and $\{x^{(j_i)}\}, \{u^{(j_i)}\}$ and $\{\lambda^{(j_i)} \}$ are bounded, it follows from \eqref{eq:ADM-opt-x} that the sequence $\{\eta^{(j_i)} \}$ is bounded. Hence, $\{\eta^{(j_i)} \}$ has a subsequence that converges. For the sake of simplification of notation, we still denote the subsequence as $\{\eta^{(j_i)} \}$ and denote $\eta$ as its limit.
%Assume that
%$\lim\limits_{i\rightarrow\infty}\eta^{(j_i)} = \eta$ without loss of generality
%(by taking a subsequence of $\{j_i\}$ if necessary).
Since $\lim_{j_i \to \infty}
\lambda^{(j_i)}-\lambda^{(j_i-1)}=0$, it follows that $\lim_{i\to\infty}\lambda^{(j_i-1)}=\bar \lambda$.
Then,
applying Proposition 2.1.5 in \citet{Clarke1990} and noting the uniform boundedness of $\bar \partial \rho_{\BaselG}(x^{(j_i)})$ (by Proposition \ref{prop:bounded_away_zero}) and \eqref{eq:cvg-con2},
we obtain from \eqref{eq:ADM-opt-x} and \eqref{equ:adm-eta} that
\begin{eqnarray}\label{eq:ADM-final-1}
\bar\lambda \in - \eta \bar \partial \rho_{\BaselG}(\bar{x}),\\
\label{eq:ADM-final-2} \eta (\rho_{\BaselG}(\bar{x}) - C_0) = 0.
\end{eqnarray}

The first-order optimality condition of \eqref{eq:ADM-y} in the
$j_i$th iteration is  \be \label{eq:ADM-opt-y} 0\in \bar \partial
 \rho(y^{(j_i)}) + \pi^{(j_i-1)} + \sigma_2 (y^{(j_i)} + \tY
u^{(j_i-1)}).\ee
Applying Proposition 2.1.5 in \citet{Clarke1990}, and taking limit
on both sides of \eqref{eq:ADM-opt-y}, we obtain
\begin{eqnarray}\label{eq:ADM-final-3}
\bar\pi \in - \bar \partial  \rho(\bar{y}) .
\end{eqnarray}

The first-order optimality condition
of \eqref{eq:ADM-u} in the $j_i$th iteration leads to \be
\label{eq:ADM-opt-u}
\tR^\top \lambda^{(j_i-1)} +  \sigma_1
\tR^\top ( x^{(j_i)} + \tR u^{(j_i)}) + \tY^\top \pi^{(j_i-1)} +  \sigma_2
\tY^\top ( y^{(j_i)} + \tY u^{(j_i)}) + \zeta^{(j_i)} = 0,\ee
where $\zeta^{(j_i)}\in\,\NU(u^{(j_i)} )$, which is the normal cone to
$ \Ur$ at $u^{(j_i)}$. It follows from \eqref{eq:ADM-opt-u} and the convergence of $\{(x^{(j_i)},y^{(j_i)},
u^{(j_i)}, \lambda^{(j_i)}, \pi^{(j_i)})\}$ that $\bar\zeta:=\lim\limits_{i
\rightarrow \infty} \zeta^{(j_i)}$ is well defined.
Since $\Ur$ is
compact and convex,  it follows from Proposition 2.4.4 in \citet{Clarke1990} that the normal
cone $\NU(u)$ coincides with the cone of normals.
Applying Proposition 2.1.5 in \citet{Clarke1990} to the cone of normals, we obtain
%Consequently, we  obtain
%%\begin{eqnarray}\label{eq:normalcone-limitation}
%$\lim\limits_{i \rightarrow \infty} \NU(u^{(j_i)} ) \subseteq
%\NU(\lim\limits_{i \rightarrow \infty} u^{(j_i)} )$. Then since $\zeta^{(j_i)}\in\,\NU(u^{(j_i)} )$, it %follows that
%\begin{eqnarray}\label{eq:normalcone-limitation-eta} \bar\zeta\in\, \,\NU(\bar{u}).  \end{eqnarray}
$ \bar\zeta\in\, \,\NU(\bar{u})$.
Taking limit on both sides of \eqref{eq:ADM-opt-u} and applying
\eqref{eq:ADM-final-1}, \eqref{eq:ADM-final-2}, \eqref{eq:ADM-final-3}
and $ \bar\zeta\in\, \,\NU(\bar{u})$, %\eqref{eq:normalcone-limitation-eta},
we  obtain \eqref{eq:kkt-30} with $u$ being $\bar u$. This completes
the proof.
\end{proof}

\bibliographystyle{plainnat}
\bibliography{port_Basel,optimization}

\end{document}